\documentclass[11pt]{article}
\usepackage[margin=1in]{geometry}
\usepackage[T1]{fontenc}
\usepackage{amsmath,amsfonts,amssymb,amsthm,mathtools}
\usepackage{mathpazo,tgpagella}
\usepackage[ruled,vlined,linesnumbered,commentsnumbered]{algorithm2e}
\SetKwInput{KwIn}{Input}
\SetKwInput{KwOut}{Output}
\SetKw{Return}{return}
\SetKwComment{Comment}{// }{}
\DontPrintSemicolon
\usepackage{enumitem,booktabs,array,microtype,placeins,needspace}
\usepackage[dvipsnames]{xcolor}
\usepackage{tikz,subcaption}
\usetikzlibrary{arrows.meta}
\usepackage{hyperref}
\usepackage{url}
\usepackage[capitalize,nameinlink]{cleveref}
\hypersetup{colorlinks=true,pdfpagemode=UseNone,citecolor=OliveGreen,
 linkcolor=NavyBlue,urlcolor=Magenta,pdfstartview=FitH,
 pdftitle={Nearly Optimal Sparse Navigable Graphs in General Metrics},
 pdfauthor={Hongyi Chen and Tianyi Zhang}}
\newtheorem{theorem}{Theorem}[section]
\newtheorem{lemma}[theorem]{Lemma}
\newtheorem{claim}[theorem]{Claim}

\theoremstyle{definition}
\newtheorem{definition}[theorem]{Definition}
\theoremstyle{remark}

\newcommand{\OPT}{\operatorname{OPT}}
\newcommand{\E}{\mathbb E}
\newcommand{\Prob}{\mathbb P}

\setlist{itemsep=3pt,topsep=4pt}

\title{Nearly Optimal Sparse Navigable Graphs in General Metrics}
\author{
Hongyi Chen\thanks{State Key Laboratory for Novel Software Technology,
Nanjing University, \href{mailto:hychener01@gmail.com}{\texttt{hychener01@gmail.com}}}
\and
Tianyi Zhang\thanks{State Key Laboratory for Novel Software Technology,
Nanjing University, \href{mailto:tianyiz25@nju.edu.cn}{\texttt{tianyiz25@nju.edu.cn}}}
}
\date{}
\begin{document}
\maketitle
\begin{abstract}
% !TEX root = ../main.tex
We study approximation algorithms for constructing sparse navigable graphs
in general metrics. Given an $n$-point metric $(P,d)$ and a parameter
$\alpha\ge1$, a directed graph on $P$ is $\alpha$-navigable if, for every pair of
distinct points $s,t\in P$, the source $s$ has an outgoing neighbor $p$
with $d(p,t)<d(s,t)/\alpha$. The objective is to minimize the number of
edges. Recent independent works of Khanna, Padaki, and Waingarten (2026) and Conway et al.\ (2026) 
give $O(\log n)$-approximation algorithms and show that an $o(\log n)$ approximation is NP-hard
in general. Before our work, the best known running time for such an
approximation for general $\alpha\ge1$ was
$\widetilde O(\min\{n\OPT_\alpha,n^{2.5}\})$, where $\OPT_\alpha$
is the minimum number of edges in an $\alpha$-navigable graph.

In this paper, we give a randomized algorithm that, with high probability,
runs in $\widetilde O(n^2)$ time and achieves an $O(\log n)$ approximation
for every $\alpha\ge1$.
In general, these guarantees are nearly optimal up to constant factors
in the approximation and polylogarithmic factors in the running time.

\end{abstract}
\thispagestyle{empty}
\clearpage

\setcounter{page}{1}
\pagestyle{plain}

% !TEX root = ../main.tex
\section{Introduction}\label{sec:introduction}

Nearest neighbor search aims to preprocess a dataset so that, given
a query point, an exact or approximate nearest neighbor
can be found efficiently. Graph-based methods approach this problem by building
a graph on the dataset points and answering queries by
following edges toward points closer to the query. This approach underlies
practical methods such as HNSW~\cite{Malkov2016HNSW},
NSG~\cite{FuEtAl19NSG}, and DiskANN~\cite{Subramanya2019DiskANN}. The graph must provide edges that
allow a search to make progress, while keeping neighborhoods small enough
to inspect efficiently.

Recent work has improved graph construction and search
guarantees~\cite{DiwanEtAl24,LuTao25,GollapudiEtAl25Sort,AlJazzaziEtAl25}.
Complementary studies examine the efficiency, accuracy, and
scalability of graph-based search in practical
settings~\cite{ZhangEtAl26Prune,AziziEtAl25,SinghEtAl21Fresh,GollapudiEtAl25Multi,AviMusco26}.

A central theoretical abstraction for this approach is
\emph{navigability}~\cite{IndykXu23}.

\begin{definition}[$\alpha$-navigability]\label{def:alpha-navigability}
Let $(P,d)$ be an $n$-point metric and let $\alpha\ge1$.
A directed graph $(P,E)$ with no self-loops is
\emph{$\alpha$-navigable} if, for every pair of distinct points
$s,t\in P$, some edge $(s,p)\in E$ satisfies
\[
                    \alpha d(p,t)<d(s,t).
\]
\end{definition}

We call $s$ the source, $t$ the target, and $p$ an intermediate point;
$p$ may be $t$ itself. For $\alpha=1$, the condition guarantees that a
greedy search can always move closer to a target in $P$ until it reaches
that target. Larger values of $\alpha$ require more progress from each
step. The analysis of Indyk and Xu~\cite{IndykXu23}
implies that, for every $0<\varepsilon\le1$, one can choose
$\alpha=O(1/\varepsilon)$ so that greedy search on an
$\alpha$-navigable graph finds a $(1+\varepsilon)$-approximate
nearest neighbor (ANN) of any query point in
$O(\log (\Delta / \varepsilon))$ search steps,
where $\Delta$ is the aspect ratio of $(P,d)$, defined as the
ratio of the largest pairwise distance to the smallest nonzero
pairwise distance. Logarithmic factors in asymptotic bounds are
understood to be at least $1$.

This motivates the following problem: given $(P,d)$ and
$\alpha\ge1$, build an $\alpha$-navigable graph with as few edges as
possible. Let $\OPT_\alpha$ denote the minimum number
of edges in such a graph. We seek an efficient algorithm whose
output size is within a small multiplicative factor of
$\OPT_\alpha$ for every input metric.

\paragraph{Previous work.}
Khanna, Padaki, and Waingarten~\cite{KhannaPadakiWaingarten26} and,
independently, Conway et al.~\cite{ConwayEtAl26} initiated
the study of approximation algorithms for this problem. 
Both formulate the problem as a Set Cover instance for each source.
For a fixed source $s$, each possible outgoing edge $(s,p)$
corresponds to the set of targets $t\in P\setminus\{s\}$ for which
$(s,p)$ satisfies the navigation requirement.
Choosing outgoing neighbors for $s$ is therefore equivalent
to selecting as few such sets as possible to cover
$P\setminus\{s\}$.
Applying greedy Set Cover to all $n$ instances yields an
$O(\log n)$ approximation in $O(n^3)$ time.
Both works also show that obtaining an $o(\log n)$ approximation
is NP-hard in general.
Thus the remaining algorithmic question is how quickly we can obtain the
logarithmic approximation.

For general $\alpha\ge1$,
Khanna, Padaki, and Waingarten~\cite{KhannaPadakiWaingarten26} give an
$\widetilde O(n\OPT_\alpha)$-time algorithm, and Conway et
al.~\cite{ConwayEtAl26} obtain
$\widetilde O(\min\{n\OPT_\alpha,n^{2.5}\})$ time. 
Here $\widetilde O$ suppresses factors polylogarithmic in $n$.
For $\alpha=1$, Conway et al.\ further obtain an
$\widetilde O(n^2)$-time construction, assuming that the distances
from each point to the other points are pairwise distinct.
Their work leaves near-quadratic construction for general $\alpha$
open. Khanna, Padaki, and Waingarten also use matrix multiplication
to obtain a \emph{bicriteria} guarantee: the
output is $\alpha$-navigable, but its size is compared with the optimum
under a stronger navigation requirement. These results leave the
following question:
\begin{quote}
\emph{Can we obtain an $O(\log n)$ approximation to the sparsest
$\alpha$-navigable graph in near-quadratic time for every $\alpha\ge1$?}
\end{quote}

\Needspace{12\baselineskip}
\subsection{Main Result}
\label{sec:main-result}
We give an affirmative answer to this question.
Even for $\alpha=1$, our result removes the distinct-distance
assumption in the near-quadratic construction of
Conway et al.~\cite{ConwayEtAl26}. Our contribution is a simple, unified algorithm to solve the coupled
Set Cover instances arising from navigable graph construction
in near-quadratic time.
\begin{theorem}\label{thm:main}
Let $(P,d)$ be an $n$-point metric space.
For every $\alpha\ge1$,
there is a randomized algorithm that constructs an
$\alpha$-navigable graph on $(P,d)$. With high probability, the algorithm runs in
$O(n^2\log^2 n)$ time and returns a graph with at most
$O(\log n)\OPT_\alpha$ edges, where $\OPT_\alpha$ is the minimum
number of edges in an $\alpha$-navigable graph on $(P,d)$.
\end{theorem}

The $O(\log n)$ approximation also holds separately for each source's
outdegree, and is optimal up to constants in general.
Moreover, even for $\alpha=1$, any randomized algorithm achieving
an $o(n)$ approximation with constant success probability requires
$\Omega(n^2)$ distance queries~\cite[Theorem~5]{KhannaPadakiWaingarten26}.
Thus our running time matches this query lower bound up to
polylogarithmic factors.

For any $0<\varepsilon\le1$, a more careful choice of
parameters and standard refinements of the analysis give a
$(1+\varepsilon)(\ln n+1)$ approximation in
$O_\varepsilon(n^2\log^2 n)$ time with high probability.
We present the simpler version for readability and briefly outline
this refinement at the end of Section~\ref{sec:overview}.

% !TEX root = ../main.tex
\subsection{Related Work}
\label{sec:related-work}

\paragraph{Bicriteria approximation.} For $\beta\ge\alpha$, let $\OPT_\beta$ denote the minimum number
of edges in a $\beta$-navigable graph on the same metric.
An $(\alpha,\beta)$-bicriteria approximation constructs an
$\alpha$-navigable graph and compares its edge count
with $\OPT_\beta$. Let $\omega$ denote
the matrix multiplication exponent.

Khanna, Padaki, and Waingarten~\cite{KhannaPadakiWaingarten26}
introduced this bicriteria formulation. For every $\alpha\ge1$ and
$0<\varepsilon<1$, their algorithm uses matrix multiplication
to construct an $\alpha$-navigable graph with
$O(\log n)\OPT_{2(1+\varepsilon)\alpha}$ edges in
$\widetilde O(n^\omega\log\Delta/\varepsilon)$ time. In the notation above,
this is an $(\alpha,2(1+\varepsilon)\alpha)$-bicriteria
$O(\log n)$ approximation.

In a follow-up work, Chen and Zhang~\cite{ChenZhang27} reduce the
running time to $\widetilde O(n^2\log\Delta)$ for constant
$\alpha\ge1$. Their algorithm achieves both an
$(\alpha,2\alpha+O(1))$-bicriteria $5$-approximation and an
$(\alpha,\alpha+O(1))$-bicriteria $O(\log n)$ approximation.
They also establish logarithmic approximation hardness for the
$(\alpha,2\alpha-\eta)$-bicriteria problem, for every constant
$\eta\in(0,1)$, under $\mathsf{P}\ne\mathsf{NP}$.

The constant-factor bicriteria guarantee is measured against
$\OPT_{2\alpha+O(1)}$, whereas Theorem~\ref{thm:main} gives an
$O(\log n)$ approximation relative to $\OPT_\alpha$.
Since $\OPT_\beta\ge\OPT_\alpha$ for $\beta>\alpha$,
a constant-factor bound relative to $\OPT_\beta$ does not
in general imply a constant-factor approximation relative
to $\OPT_\alpha$.
Our running time is also independent of the aspect ratio
$\Delta$.
Table~\ref{tab:related-objectives} summarizes the results.

\begin{table}[t]
\centering
\small
\setlength{\tabcolsep}{4pt}
\renewcommand{\arraystretch}{1.15}
\caption{Constructions of sparse $\alpha$-navigable graphs in general
metrics. Here $\Delta$ is the aspect ratio and $\omega$ is the matrix
multiplication exponent. The $\alpha=1$ result of Conway et al.\ assumes
distinct distances from each target. The matrix-multiplication
bound of Khanna et al.\ uses $0<\varepsilon<1$.}
\label{tab:related-objectives}
\begin{tabular}{@{}>{\raggedright\arraybackslash}p{0.27\textwidth}
>{\raggedright\arraybackslash}p{0.36\textwidth}
>{\raggedright\arraybackslash}p{0.33\textwidth}@{}}
\toprule
\textbf{Reference} & \textbf{Edge bound} & \textbf{Running time}\\
\midrule
Khanna et al.~\cite{KhannaPadakiWaingarten26}
& $O(\log n)\OPT_\alpha$
& $\widetilde O(n\OPT_\alpha)$\\[4pt]
Khanna et al.~\cite{KhannaPadakiWaingarten26}
& $O(\log n)\OPT_{2(1+\varepsilon)\alpha}$
& $\widetilde O(n^\omega\log\Delta/\varepsilon)$\\[4pt]
Conway et al.~\cite{ConwayEtAl26}
& $O(\log n)\OPT_\alpha$
& $\widetilde O(\min\{n\OPT_\alpha,n^{2.5}\})$\\[4pt]
Conway et al.~\cite{ConwayEtAl26}\newline ($\alpha=1$)
& $O(\log n)\OPT_1$
& $\widetilde O(n^2)$\\[4pt]
Chen--Zhang~\cite{ChenZhang27}
& $5\OPT_{2\alpha+O(1)}$
& $\widetilde O(n^2\log\Delta)$\\[4pt]
Chen--Zhang~\cite{ChenZhang27}
& $O(\log n)\OPT_{\alpha+O(1)}$
& $\widetilde O(n^2\log\Delta)$\\[4pt]
\textbf{This paper}
& $O(\log n)\OPT_\alpha$
& $\widetilde O(n^2)$\\
\bottomrule
\end{tabular}
\end{table}

\paragraph{Euclidean spaces.}
For fixed dimension $d$, Arya and Mount~\cite{AryaMount93} and
Fu et al.~\cite{FuEtAl19NSG} construct $1$-navigable graphs
on $n$ points in $\mathbb R^d$ under the Euclidean distance,
with $n2^{O(d)}$ edges in $O(n^2)$ and $O(n^2\log n)$ time,
respectively, assuming distinct distances.
For every $\alpha\ge1$, Indyk and Xu~\cite{IndykXu23}
construct an $\alpha$-navigable graph with
$O(n(2\alpha)^{O(d)}\log\Delta)$ edges in $O(n^3)$ time.

High dimension can force dense navigable graphs. For every fixed
$\gamma>0$, Diwan et al.~\cite{DiwanEtAl24} give Euclidean point
sets in $O_\gamma(\log n)$ dimensions requiring
$\Omega(n^{3/2-\gamma})$ edges for $1$-navigability, and point sets
in $O(\log^3 n)$ dimensions requiring
$\Omega(n^{3/2}/\log n)$ edges. Waizer and
Danieli~\cite{WaizerDanieli26} obtain an $\Omega(n^{3/2})$ lower
bound in $O(\sqrt n\log n)$ dimensions with pairwise distinct
distances. Together with the $O(n^{3/2})$ construction of Conway
et al.~\cite{ConwayEtAl26}, this gives a tight worst-case edge
bound under the distinct-distance assumption. With distance ties,
an equilateral metric requires all $n(n-1)$ directed edges.
Diwan et al.~\cite{DiwanEtAl24} also give
$O(\log n)$-dimensional instances requiring maximum outdegree
$\Omega(n)$ for $1$-navigability. For every fixed $\alpha>1$,
Al-Jazzazi et al.~\cite{AlJazzaziEtAl25} give $n$ points in
$O(\log n)$-dimensional Euclidean space, with pairwise distinct
distances, for which every $\alpha$-navigable graph requires all
$n(n-1)$ directed edges.

Other graph constructions in Euclidean space focus on efficient
ANN search rather than $\alpha$-navigability.
For $\varepsilon\in(0,1)$, Lu and
Tao~\cite{LuTao25} construct graphs with
$O(n\varepsilon^{-O(d)})$ edges in
$n\operatorname{polylog}(n\Delta)$ time with high probability,
supporting $(1+\varepsilon)$-ANN queries in
$\operatorname{polylog}(n\Delta)$ time.
Har-Peled, Raichel, and Robson~\cite{HarPeledRaichelRobson26}
obtain graphs with $O(n\varepsilon^{-d})$ edges supporting
$(1+\varepsilon)$-ANN queries in
$O(\varepsilon^{-d}\log\Delta)$ time.

\paragraph{Almost navigability.}
Avi and Musco~\cite{AviMusco26} relax $1$-navigability by requiring
that, for at least a $(1-\tau)$ fraction of each source's other
targets, some outgoing neighbor is closer to the target than
the source.
Assuming pairwise distinct distances, for $0<\tau,\xi<1$,
their randomized algorithm constructs a graph with $O(n/\tau)$
edges in $O(n\log(n/\xi)/\tau)$ time with success probability
at least $1-\xi$, assuming constant-time distance access.
Waizer and Danieli~\cite{WaizerDanieli26} show that the worst-case
edge count over Euclidean point sets with pairwise distinct distances is
$\Theta(n\min\{\tau^{-1},\sqrt n\})$ for all sufficiently large
$n$ and all sufficiently small $\tau>0$.
These results characterize the worst-case edge count for
almost navigability. Our algorithm approximates the optimum
on each input metric while ensuring $\alpha$-navigability
for every ordered pair, without assuming distinct distances.

% !TEX root = ../main.tex
\section{Technical Overview}\label{sec:overview}

For a fixed source $s$, choosing outgoing neighbors can be
formulated as a Set Cover problem, following previous
work~\cite{ConwayEtAl26,KhannaPadakiWaingarten26}.
For each possible intermediate point $p\ne s$, define
the set of targets for which the edge $(s,p)$ satisfies the navigation requirement:
\[
             S_s(p)=\{t\in P\setminus\{s\}:\alpha d(p,t)<d(s,t)\}.
\]
Choosing outgoing neighbors that satisfy the navigation requirement
for $s$ is equivalent to covering $P\setminus\{s\}$ with the
corresponding sets $S_s(p)$.
Let $k_s$ denote the minimum number of sets needed, equivalently
the minimum feasible outdegree of $s$.

We say that $p$ \emph{covers} $t$ when $t\in S_s(p)$. 
Let $U$ denote the \emph{uncovered targets}, those not
covered by any selected intermediate point.
Initially, $U=P\setminus\{s\}$.
Adding $p$ to the outgoing neighborhood inserts $(s,p)$
and updates $U\leftarrow U\setminus S_s(p)$.

An \emph{incidence} is a pair $(p,t)$ with $t\in S_s(p)$.
These incidences define a bipartite \emph{incidence graph}:
each left vertex corresponds to an intermediate point $p\ne s$
and represents the set $S_s(p)$, while each right vertex
corresponds to a target $t\ne s$.
The two vertices are adjacent exactly when $t\in S_s(p)$.
The degree of a target is its \emph{frequency}, the number
of sets containing it; the degree of an intermediate point
$p$ is the set size $|S_s(p)|$.

The graph has at most $n$ vertices on each side and can contain
$\Theta(n^2)$ edges.
Constructing it explicitly and running greedy,
which repeatedly selects a set covering the most uncovered
targets, takes $O(n^2)$ time per source and $O(n^3)$ time overall.
Our goal is to preserve the $O(\log n)$ approximation
while reducing the total running time to near-quadratic.

\paragraph{Limiting target frequency by random pivots.}

Suppose that we know a degree estimate $D$ satisfying
$k_s\le D\le2k_s$. Since we seek an $O(\log n)$ approximation,
we can afford $O(D\log n)$ outgoing edges.
Let $Q$ be the set of points obtained from $O(D\log n)$
independent uniform samples from $P$; we call these points
\emph{pivots}.
Add the edges $(s,q)$ for $q\in Q\setminus\{s\}$.

We show that after adding the pivot edges, every uncovered
target has frequency $O(n/D)$ with high probability.
\Cref{fig:random-pivot-prefix-overview} illustrates the argument.

% !TEX root = ../main.tex
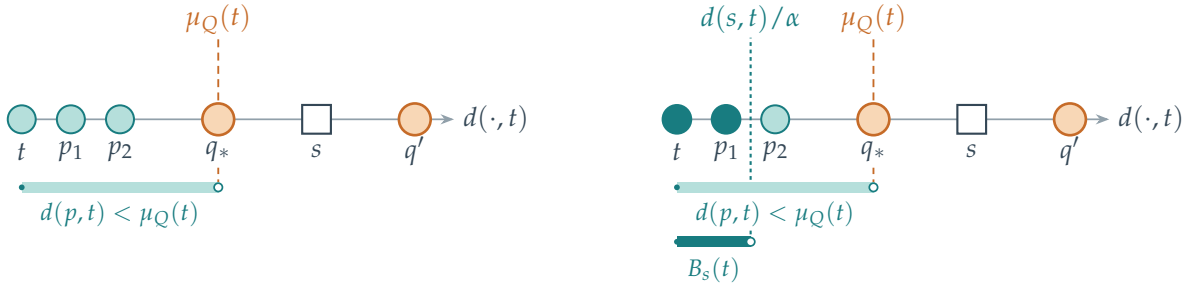
\begingroup
\definecolor{qpInk}{HTML}{344756}
\definecolor{qpGray}{HTML}{91A0AA}
\definecolor{qpOrange}{HTML}{C66B28}
\definecolor{qpOrangeFill}{HTML}{F8D7B6}
\definecolor{qpTeal}{HTML}{187C80}
\definecolor{qpTealFill}{HTML}{B5DFDB}
\tikzset{
  qp point/.style={circle,draw=qpInk,fill=white,
    line width=.75pt,minimum size=3.7mm,inner sep=0pt},
  qp prefix point/.style={qp point,draw=qpTeal,fill=qpTealFill},
  qp intermediate/.style={qp point,draw=qpTeal,fill=qpTeal,
    line width=1pt},
  qp pivot/.style={qp point,draw=qpOrange,fill=qpOrangeFill,
    line width=1pt,minimum size=4.2mm},
  qp source/.style={rectangle,draw=qpInk,fill=white,
    line width=.8pt,minimum size=3.8mm,inner sep=0pt},
  qp mu guide/.style={draw=qpOrange,densely dashed,line width=.7pt},
  qp threshold guide/.style={draw=qpTeal,dash pattern=on 1.4pt off 1.4pt,
    line width=.8pt},
  qp label/.style={font=\small,text=qpInk},
}
% Six representative distance positions, shared by both panels:
% t=0, p1=1, p2=2, q_*=4, s=6, q'=8, scaled by .65 cm.
% These distances are realized by six points on the real line.
% Let Q={q_*,q'}, alpha=4: d(s,t)/alpha=6/4=1.5 <= mu_Q(t)=4.
% Points satisfying d(p,t)<mu_Q(t) are {t,p1,p2}; B_s(t)={t,p1}, since t may be its own
% intermediate point and only source s is excluded.
% The displayed Q and point counts are schematic, not the complete
% prescribed O(D log n) pivot sample for this six-point metric.
\newcommand{\qpFrame}{%
  \path[use as bounding box] (-.35,-.80) rectangle (6.85,3.0);}
\newcommand{\qpAxis}{%
  \draw[draw=qpGray,line width=.65pt,-{Stealth[length=1.8mm]}]
    (-.12,1.40) -- (5.72,1.40);
  \node[qp label,anchor=west,inner sep=1pt] at (5.80,1.40) {$d(\cdot,t)$};}
\newcommand{\qpMu}{%
  \draw[qp mu guide] (2.60,1.65) -- (2.60,2.48);
  \draw[qp mu guide] (2.60,.50) -- (2.60,.76);
  \node[font=\small,text=qpOrange] at (2.60,2.72) {$\mu_Q(t)$};}
\newcommand{\qpPrefixBand}{%
  \fill[qpTealFill] (0,.43) rectangle (2.60,.57);
  \fill[qpTeal] (0,.50) circle[radius=.038];
  \draw[draw=qpTeal,fill=white,line width=.65pt]
    (2.60,.50) circle[radius=.055];
  \node[font=\footnotesize,text=qpTeal] at (1.3,.14) {$d(p,t)<\mu_Q(t)$};}
\newcommand{\qpCommonNodes}{%
  \node[qp pivot] at (2.60,1.40) {};
  \node[qp source] at (3.90,1.40) {};
  \node[qp pivot] at (5.20,1.40) {};
  \foreach \x/\lab in {0/t,.65/p_1,1.30/p_2,2.60/q_*,3.90/s,5.20/q'} {
    \node[qp label] at (\x,.98) {$\lab$};
  }}

\begin{figure}[!htbp]
\captionsetup{font=small,labelfont=bf,labelsep=period}
\captionsetup[subfigure]{font=small,labelfont=bf,labelsep=space,
  justification=raggedright,singlelinecheck=false,skip=8pt}
\centering
\begin{subfigure}[t]{.475\textwidth}
  \centering
  \begin{tikzpicture}
    \qpFrame
    \qpAxis
    \qpMu
    \qpPrefixBand
    \foreach \x in {0,.65,1.30}
      \node[qp prefix point] at (\x,1.40) {};
    \qpCommonNodes
  \end{tikzpicture}
  \caption{\textbf{Few points closer than the nearest pivot.}
  Here points are ordered from left to right by their distance from $t$ for illustration.
  The nearest pivot $q_*$ has distance $\mu_Q(t)=d(q_*,t)$ from $t$.
  With high probability, only $O(n/D)$ points satisfy
  $d(p,t)<\mu_Q(t)$.}
  \label{fig:random-pivot-prefix}
\end{subfigure}\hfill
\begin{subfigure}[t]{.475\textwidth}
  \centering
  \begin{tikzpicture}
    \qpFrame
    \qpAxis
    \qpMu
    \draw[qp threshold guide] (.975,.60) -- (.975,2.48);
    \draw[qp threshold guide] (.975,-.22) -- (.975,-.04);
    \node[font=\small,text=qpTeal] at (.975,2.72) {$d(s,t)/\alpha$};
    \qpPrefixBand
    \fill[qpTeal] (0,-.29) rectangle (.975,-.15);
    \fill[qpTeal] (0,-.22) circle[radius=.038];
    \draw[draw=qpTeal,fill=white,line width=.65pt]
      (.975,-.22) circle[radius=.055];
    \node[font=\footnotesize,text=qpTeal] at (.4875,-.59) {$B_s(t)$};
    \foreach \x in {0,.65}
      \node[qp intermediate] at (\x,1.40) {};
    \node[qp prefix point] at (1.30,1.40) {};
    \qpCommonNodes
  \end{tikzpicture}
  \caption{\textbf{Uncovered targets have few choices.}
  If $t$ remains uncovered after adding edges from $s$ to $Q\setminus\{s\}$,
  then $d(s,t)/\alpha\le\mu_Q(t)$.
  Every intermediate point that covers $t$ satisfies $d(p,t)<d(s,t)/\alpha$.
  Denote these points by $B_s(t)$; all satisfy $d(p,t)<\mu_Q(t)$.}
  \label{fig:random-pivot-candidates}
\end{subfigure}

\vspace{4mm}
{\footnotesize
\tikz[baseline=-.6ex]\node[qp pivot,minimum size=2.8mm] {};
\ Sampled pivots
\qquad
\tikz[baseline=-.6ex]\node[qp prefix point,minimum size=2.8mm] {};
\ Points with $d(p,t)<\mu_Q(t)$
\qquad
\tikz[baseline=-.6ex]\node[qp intermediate,minimum size=2.8mm] {};
\ Intermediate points covering $t$
\qquad
\tikz[baseline=-.6ex]\node[qp source,minimum size=2.8mm] {};
\ Source
}

\caption{Random pivots bound target frequency.}
\label{fig:random-pivot-prefix-overview}
\end{figure}
\endgroup

Let $\mu_Q(t)$ be the distance from $t$ to its nearest pivot:
\[
                      \mu_Q(t)=\min_{q\in Q}d(q,t).
\]
If $t$ remains uncovered, no pivot covers it, so
$\mu_Q(t)\ge d(s,t)/\alpha$.
Every intermediate point $p$ covering $t$ must therefore satisfy
\[
    d(p,t)<\frac{d(s,t)}{\alpha}\le\mu_Q(t).
\]
Thus it suffices to bound the number of points strictly closer
to $t$ than its nearest pivot.

For a fixed target $t$, hitting one of its $\lceil n/D\rceil$
nearest points, with ties broken arbitrarily, ensures that
fewer than $\lceil n/D\rceil$ points lie strictly closer to $t$
than its nearest pivot.
Each draw hits these points with probability at least $1/D$,
so $O(D\log n)$ independent draws hit one with high probability.
Thus, with high probability, if $t$ remains uncovered,
its frequency is $O(n/D)$.
A union bound over all $n$ targets gives this frequency bound
for every uncovered target simultaneously with high probability.

The event that $Q$ hits one of these nearest points for every
target does not depend on the source.
Hence the frequency bound holds simultaneously for all sources,
and sources with the same degree estimate $D$ can share $Q$.

\paragraph{Sampled greedy coverage.}
In each iteration, we aim to choose at most $D$ additional
intermediate points covering at least half of the current
uncovered targets $U$. 
We sample targets from $U$ and run greedy on the sampled
Set Cover instance to choose these points.
We then check their coverage on $U$ exactly.
If they cover at least half of $U$, we add the
corresponding edges $(s,p)$ and remove the covered targets from $U$.
Otherwise, we discard the proposed points and report failure.

At the start of each iteration, draw a uniform sample
$T$ of $O(D\log n)$ targets from $U$ without replacement.
If $|U|$ is smaller than the prescribed sample size, take $T=U$.
This defines a \emph{sampled Set Cover instance} with targets
$T$ and sets $S_s(p)\cap T$ for $p\ne s$.
Its incidence graph retains all left vertices and only
the target vertices in $T$.

Run greedy on the sampled instance for at most $D$ steps,
each time choosing a restricted set covering the most
remaining sampled targets.
Stop earlier if all sampled targets are covered.
An optimal neighborhood of $s$ provides at most $k_s$ sets
covering $T$.
At each step, one of these sets covers at least a $1/k_s$
fraction of the remaining sampled targets, and greedy
covers at least as many.
Since $D\ge k_s$, the covered fraction of $T$ is therefore
at least
\[
    1-(1-1/k_s)^D\ge1-e^{-1}>5/8.
\]
We only need this constant coverage, rather than a cover
of the entire sample.

We use a standard sampling argument%
~\cite{HarPeledSharir11,HarPeledEtAl16}.
Fix any collection of at most $D$ intermediate points
covering less than half of $U$.
Hoeffding's inequality bounds the probability that it covers
at least $5/8$ of $T$ by $e^{-|T|/32}$.
There are at most $n^D$ such collections.
Thus $O(D\log n)$ sampled targets suffice, by a union bound,
to ensure that every collection covering at least $5/8$ of $T$
also covers at least half of $U$, with high probability.
If $T=U$, this implication holds deterministically.
This guarantee holds simultaneously for all such choices
on the same sample $T$.
It therefore applies to greedy's choices, even though they
are made after examining $T$.
Consequently, the proposed points cover at least half of $U$
with high probability.

For each target $t$, we sort $P$ by distance to $t$ once
and reuse this order for all sources.
By the frequency bound, each sampled target has degree
$O(n/D)$, so the sampled incidence graph has
$O(D\log n)\cdot O(n/D)=O(n\log n)$ edges.
For each $t\in T$, we scan its distance order until
$\alpha d(p,t)\ge d(s,t)$ to enumerate its incident edges.
We implement greedy using buckets indexed by the number
of sampled targets a set still covers, updating the affected
sets whenever a sampled target is covered.
Each incidence is processed $O(1)$ times, so constructing
the sampled incidence graph and running greedy take
$O(n\log n)$ time in total.

Each successful iteration halves $|U|$, so $O(\log n)$
such iterations use $O(D\log n)$ neighbors in total.
However, we must still find the targets in $U$ covered by
the proposed points, both to check their coverage exactly
and to update $U$.
Since a selected set may cover $O(n)$ targets, directly
processing up to $D$ selected sets can cost
$O(Dn)$ time per iteration.

\paragraph{Bounding set sizes with pivots.}

We obtain a second pivot set $R$ of size $O(D\log n)$
using a variant of Thorup and Zwick's center-selection
method~\cite[Theorem~3.1]{ThorupZwick01}.
Define the nearest-pivot distances and the lists by
\[
    \mu_R(t)=\min_{r\in R}d(r,t),
    \qquad
    I_R(p)=\{t\in P:d(p,t)<\mu_R(t)\}.
\]
The list $I_R(p)$ contains the targets for which $p$ is closer
than every pivot in $R$. These lists coincide with their clusters.
We require
\[
    |I_R(p)|\le n/D
    \qquad\text{for every }p\in P.
\]
We first explain how these lists eliminate the update bottleneck,
then show how to construct $R$ and the lists.

After adding the edges $(s,r)$ for $r\in R\setminus\{s\}$,
every target $t$ remaining in $U$ satisfies
$\mu_R(t)\ge d(s,t)/\alpha$.
Thus an intermediate point $p$ can cover $t$ only if
\[
    d(p,t)<\frac{d(s,t)}{\alpha}\le\mu_R(t),
\]
and consequently
\[
    S_s(p)\cap U\subseteq I_R(p).
\]

For the at most $D$ intermediate points proposed by greedy,
we scan their lists $I_R(p)$ and check which listed targets
belong to $U$ and satisfy the navigation inequality.
Marking each covered target only once gives the exact coverage
in $O(D\cdot n/D)=O(n)$ time.
If at least half of $U$ is covered, we add the proposed edges
and remove the marked targets from $U$; otherwise, we leave
$U$ unchanged.

Since $R$, $\mu_R$, and the lists $I_R(p)$ for $p\in P$
do not depend on the source, they can be shared by all sources
with the same degree estimate $D$.
For each source $s$, we add both $Q\setminus\{s\}$ and
$R\setminus\{s\}$ to its neighborhood before the greedy
iterations, contributing $O(D\log n)$ neighbors in total.
Using $\mu_Q(t),\mu_R(t)$, we can filter the current
$U$ in $O(n)$ time, retaining exactly the targets satisfying
\[
    \frac{d(s,t)}{\alpha}
    \le \min\{\mu_Q(t),\mu_R(t)\}.
\]
The set $Q$ bounds target frequencies, allowing us to construct
the sampled incidence graph and run greedy in $O(n\log n)$ time.
The lists obtained from $R$ support exact validation and
updates in $O(n)$ time.
Together, they give $O(n\log n)$ time per iteration.
\Cref{fig:residual-sampling-overview} illustrates this combination
of pivots, sampled greedy coverage, and exact validation.
% !TEX root = ../main.tex
\begingroup
\definecolor{riInk}{HTML}{344756}
\definecolor{riEdge}{HTML}{82929D}
\definecolor{riGhost}{HTML}{E5E9EC}
\definecolor{riOrange}{HTML}{C66B28}
\definecolor{riOrangeFill}{HTML}{F8D7B6}
\definecolor{riTeal}{HTML}{187C80}
\definecolor{riTealFill}{HTML}{B5DFDB}
\tikzset{
  ri intermediate/.style={rectangle,draw=riInk,fill=white,
    line width=.75pt,minimum size=3.4mm,inner sep=0pt},
  ri target/.style={circle,draw=riInk,fill=white,
    line width=.75pt,minimum size=3.8mm,inner sep=0pt},
  ri selected/.style={ri intermediate,draw=riOrange,fill=riOrangeFill,
    line width=1.15pt,minimum size=4.2mm},
  ri covered/.style={ri target,draw=riTeal,fill=riTealFill,
    line width=1pt,minimum size=4.2mm},
  ri sampled/.style={ri target,draw=riTeal,line width=1pt,
    minimum size=4.2mm},
  ri ghost/.style={ri target,draw=riGhost,fill=riGhost!35,
    line width=.6pt},
  ri edge/.style={draw=riEdge,line width=.65pt},
  ri dense edge/.style={draw=riEdge!85,line width=.55pt},
  ri faint edge/.style={draw=riGhost,line width=.55pt},
  ri cover edge/.style={draw=riTeal,line width=1.1pt},
  ri heading/.style={font=\small,text=riInk,align=center},
}
% Small, realizable metric example: P={s,p_1,...,p_8}, alpha=2,
% n=9, residual degree bound 2, combined pivots {s,p_1,p_8}.
% d(s,p)=2 for residual points p_2,...,p_7, and 2.5 otherwise.
% Distances between distinct intermediate points are .9 for pairs {2,3} and {5,6},
% and 1.1 otherwise. All metric triangles and both shared certificate
% bounds hold on the full point set, including s.
% Adding (Q union R)\{s} removes exactly targets 1 and 8.
\def\riEdges{2/2,2/3,3/2,3/3,4/4,5/5,5/6,6/5,6/6,7/7}
\def\riSampleEdges{2/2,2/3,3/2,3/3,5/5,6/5,7/7}
\def\riSampleCoverEdges{2/2,2/3,5/5,7/7}
\def\riFullCoverEdges{2/2,2/3,5/5,5/6,7/7}
\newcommand{\riLeftCoordinates}{%
  \foreach \p in {1,...,8}
    \coordinate (L\p) at (0,{3.6-(\p-1)*3.6/7});}
\newcommand{\riResidualCoordinates}{%
  \riLeftCoordinates
  \foreach \t in {2,...,7}
    \coordinate (R\t) at (4.4,{3.6-(\t-2)*3.6/5});}
\newcommand{\riLeftNodes}{%
  \foreach \p in {1,...,8} \node[ri intermediate] at (L\p) {};}
\newcommand{\riSelectedNodes}{%
  \foreach \p in {2,5,7} {
    \node[ri selected] at (L\p) {};
    \node[anchor=east,font=\small,text=riOrange,xshift=-3.5mm]
      at (L\p) {$p_{\p}$};}}
\newcommand{\riPanelBounds}{%
  \path[use as bounding box] (-1.55,-.20) rectangle (5.55,4.55);}
\newcommand{\riLabels}[1]{%
  \node[ri heading] at (0,4.22) {Intermediate points\\$P\setminus\{s\}$};
  \node[ri heading] at (4.4,4.22) {#1};}

\begin{figure}[!p]
\captionsetup{font=small,labelfont=bf,labelsep=period}
\captionsetup[subfigure]{font=small,labelfont=bf,labelsep=space,
  justification=raggedright,singlelinecheck=false,skip=9pt}
\centering
\begin{subfigure}[t]{.475\textwidth}
  \centering
  \begin{tikzpicture}
    \riPanelBounds
    \riLeftCoordinates
    \foreach \t in {1,...,8}
      \coordinate (R\t) at (4.4,{3.6-(\t-1)*3.6/7});
    \foreach \t in {1,8}
      \foreach \p in {1,...,8}
        \draw[ri dense edge] (L\p) -- (R\t);
    \foreach \p/\t in \riEdges
      \draw[ri dense edge] (L\p) -- (R\t);
    \riLeftNodes
    \foreach \t in {1,...,8} \node[ri target] at (R\t) {};
    \riLabels{Targets\\$P\setminus\{s\}$}
  \end{tikzpicture}
  \caption{\textbf{Intermediate points and targets.}
  Fix a source $s$.
  Left vertices represent intermediate points $p\ne s$,
  and right vertices represent targets $t\ne s$.
  An edge $(p,t)$ means $t\in S_s(p)$, equivalently
  $\alpha d(p,t)<d(s,t)$.}
  \label{fig:residual-before-centers}
\end{subfigure}\hfill
\begin{subfigure}[t]{.475\textwidth}
  \centering
  \begin{tikzpicture}
    \riPanelBounds
    \riResidualCoordinates
    \foreach \p/\t in \riEdges
      \draw[ri edge] (L\p) -- (R\t);
    \riLeftNodes
    \foreach \t in {2,...,7} \node[ri target] at (R\t) {};
    \riLabels{Uncovered targets\\$U$}
  \end{tikzpicture}
  \caption{\textbf{After adding the pivot edges.}
  Add edges from $s$ to $(Q\cup R)\setminus\{s\}$
  and remove covered targets, leaving $U$.
  With high probability, the remaining graph has maximum degree
  $O(n/D)$ on both sides: $Q$ bounds target frequencies,
  and $R$ bounds set sizes.
  This graph is kept implicit.}
  \label{fig:residual-after-centers}
\end{subfigure}

\vspace{8mm}

\begin{subfigure}[t]{.475\textwidth}
  \centering
  \begin{tikzpicture}
    \riPanelBounds
    \riResidualCoordinates
    \foreach \p/\t in \riEdges
      \draw[ri faint edge] (L\p) -- (R\t);
    \foreach \p/\t in \riSampleEdges
      \draw[ri edge] (L\p) -- (R\t);
    \foreach \p/\t in \riSampleCoverEdges
      \draw[ri cover edge] (L\p) -- (R\t);
    \riLeftNodes
    \foreach \t in {2,...,7} \node[ri ghost] at (R\t) {};
    \foreach \t in {2,3,5,7} \node[ri sampled] at (R\t) {};
    \foreach \t in {2,3,5,7} \node[ri covered] at (R\t) {};
    \riSelectedNodes
    \riLabels{Sampled targets\\$T$}
  \end{tikzpicture}
  \caption{\textbf{Sampled greedy coverage.}
  Sample targets $T\subseteq U$ uniformly without replacement.
  Keep all intermediate points and only target vertices in $T$.
  Construct the sampled graph explicitly and run greedy
  for at most $D$ choices, obtaining $S$ (orange).
  Unsampled targets and their edges are faded.}
  \label{fig:residual-sampled-greedy}
\end{subfigure}\hfill
\begin{subfigure}[t]{.475\textwidth}
  \centering
  \begin{tikzpicture}
    \riPanelBounds
    \riResidualCoordinates
    \foreach \p/\t in \riEdges
      \draw[ri edge] (L\p) -- (R\t);
    \foreach \p/\t in \riFullCoverEdges
      \draw[ri cover edge] (L\p) -- (R\t);
    \riLeftNodes
    \foreach \t in {2,...,7} \node[ri target] at (R\t) {};
    \foreach \t in {2,3,5,6,7} \node[ri covered] at (R\t) {};
    \riSelectedNodes
    \riLabels{Targets in $U$}
  \end{tikzpicture}
  \caption{\textbf{Exact validation on $U$.}
  Check the same selected points $S$ against all targets in $U$
  by scanning their lists $I_R(p)$.
  If at least half of $U$ is covered, add $S$ to the neighborhood
  and remove the covered targets; otherwise, leave $U$ unchanged.}
  \label{fig:residual-lifted-coverage}
\end{subfigure}

\vspace{5mm}
{\footnotesize
\tikz[baseline=-.6ex]\node[ri selected,minimum size=3mm] {};
\ Selected intermediate points $S$
\qquad
\tikz[baseline=-.6ex]\node[ri covered,minimum size=3mm] {};
\ Covered targets
\par\vspace{2mm}

\tikz[baseline=-.6ex]\node[ri sampled,minimum size=3mm] {};
\ Sampled, not covered by $S$
\qquad
\tikz[baseline=-.6ex]\node[ri ghost,minimum size=3mm] {};
\ Unsampled in (c)
}

\caption{Pivots and sampled greedy coverage.}
\label{fig:residual-sampling-overview}
\end{figure}
\endgroup

\clearpage

\paragraph{Constructing \texorpdfstring{$R$}{R} by adaptive sampling.}

Call a point $p$ \emph{satisfied} when its current list
$I_R(p)$ has size at most $n/D$.
Start with pivots $R=\varnothing$ and set $A=P$, interpreting
$\mu_\varnothing(t)=+\infty$.
While $|A|>2D$, perform a round: draw $2D$ points
independently and uniformly from $A$, add them to $R$,
and remove from $A$ every point that is now satisfied.
When $|A|\le2D$, add all remaining points of $A$ to $R$
and finish.

Adding pivots can only shrink the lists, so every point
removed from $A$ remains satisfied.
Each point added to $R$ has $I_R(p)=\varnothing$, so the final
insertion satisfies all remaining points.
Thus every point is satisfied when the construction finishes.
\Cref{fig:adaptive-pivot-overview} illustrates one round.

% !TEX root = ../main.tex
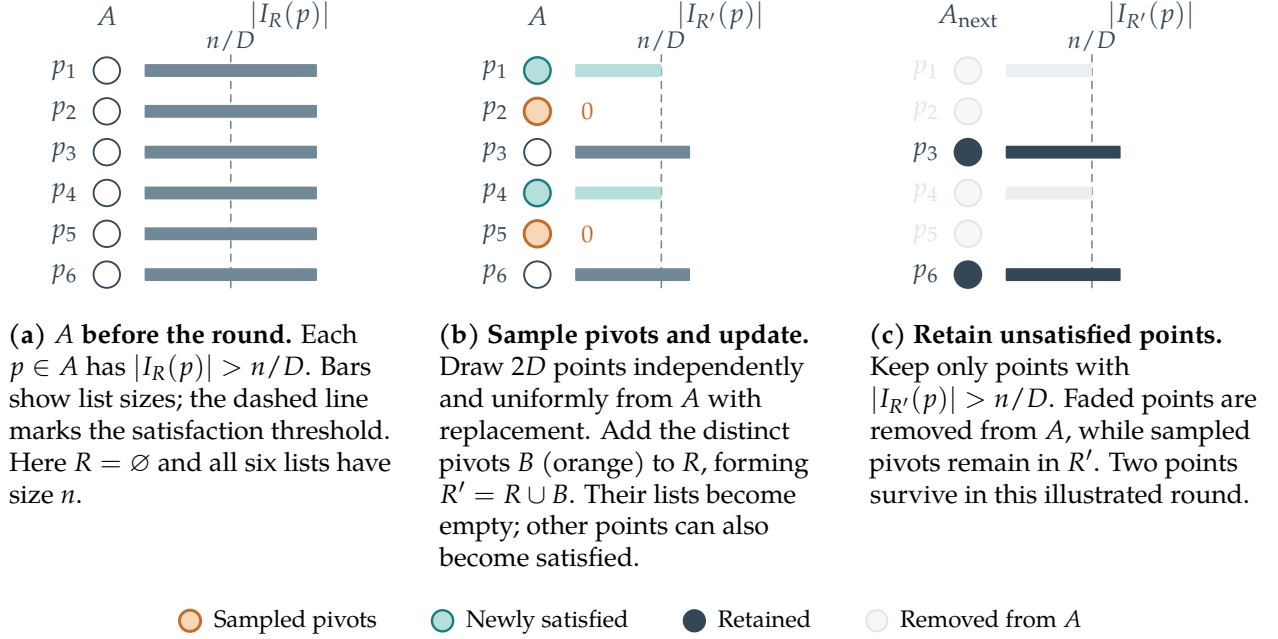
\begingroup
\definecolor{apInk}{HTML}{344756}
\definecolor{apBar}{HTML}{718897}
\definecolor{apGhost}{HTML}{DFE5E9}
\definecolor{apOrange}{HTML}{C66B28}
\definecolor{apOrangeFill}{HTML}{F8D7B6}
\definecolor{apTeal}{HTML}{187C80}
\definecolor{apTealFill}{HTML}{B5DFDB}
\tikzset{
  ap point/.style={circle,draw=apInk,fill=white,
    line width=.7pt,minimum size=3.5mm,inner sep=0pt},
  ap pivot/.style={ap point,draw=apOrange,fill=apOrangeFill,
    line width=1.1pt},
  ap satisfied/.style={ap point,draw=apTeal,fill=apTealFill,
    line width=.9pt},
  ap retained/.style={ap point,draw=apInk,fill=apInk,
    line width=.8pt},
  ap faded/.style={ap point,draw=apGhost,fill=apGhost!30},
  ap label/.style={anchor=east,font=\small,text=apInk},
  ap threshold/.style={draw=apInk!65,densely dashed,line width=.6pt},
}
% Exact metric-realizable illustration, not a generic probability claim:
% P={p1,...,p6}, n=6, D=2, threshold n/D=3, R initially empty.
% Distinct distances are .9 for undirected pairs {1,3}, {1,6},
% {3,4}, {3,6}, {4,6}; all other distinct distances are 1.1.
% Four independent draws can give (p2,p5,p2,p5), so B={p2,p5}.
% List sizes change from (6,6,6,6,6,6) to (3,0,4,3,0,4).
% Points p1,p4 meet the threshold; pivots p2,p5 have empty lists.
% A_next={p3,p6}. Monotonicity keeps removed points satisfied.
\newcommand{\apFrame}[2]{%
  \path[use as bounding box] (-.67,-.33) rectangle (3.20,3.63);
  \node[font=\small,text=apInk] at (0,3.43) {$#1$};
  \node[font=\small,text=apInk] at (2.42,3.43) {$#2$};
  \draw[ap threshold] (1.64,-.18) -- (1.64,2.95);
  \node[font=\footnotesize,text=apInk] at (1.64,3.10) {$n/D$};}
\newcommand{\apCoordinates}{%
  \foreach \p in {1,...,6}
    \coordinate (P\p) at (0,{2.70-(\p-1)*.54});}
\newcommand{\apBarDraw}[3]{%
  \pgfmathsetmacro{\apY}{2.70-(#1-1)*.54}
  \fill[#3,rounded corners=.4pt]
    (.50,{\apY-.085}) rectangle ({.50+#2*.38},{\apY+.085});}
\newcommand{\apNormalLabels}{%
  \foreach \p in {1,...,6}
    \node[ap label,xshift=-2.5mm] at (P\p) {$p_{\p}$};}

\begin{figure}[!ht]
\captionsetup{font=small,labelfont=bf,labelsep=period}
\captionsetup[subfigure]{font=small,labelfont=bf,labelsep=space,
  justification=raggedright,singlelinecheck=false,skip=8pt}
\centering
\begin{subfigure}[t]{.31\textwidth}
  \centering
  \begin{tikzpicture}
    \apFrame{A}{|I_R(p)|}
    \apCoordinates
    \foreach \p in {1,...,6} {
      \apBarDraw{\p}{6}{apBar}
      \node[ap point] at (P\p) {};
    }
    \apNormalLabels
  \end{tikzpicture}
  \caption{\textbf{$A$ before the round.}
  Each $p\in A$ has $|I_R(p)|>n/D$.
  Bars show list sizes; the dashed line marks the satisfaction
  threshold. Here $R=\varnothing$ and all six lists have size $n$.}
  \label{fig:adaptive-pivot-before}
\end{subfigure}\hfill
\begin{subfigure}[t]{.31\textwidth}
  \centering
  \begin{tikzpicture}
    \apFrame{A}{|I_{R'}(p)|}
    \apCoordinates
    \foreach \p in {1,4} {
      \apBarDraw{\p}{3}{apTealFill}
      \node[ap satisfied] at (P\p) {};
    }
    \foreach \p in {3,6} {
      \apBarDraw{\p}{4}{apBar}
      \node[ap point] at (P\p) {};
    }
    \foreach \p in {2,5} {
      \node[ap pivot] at (P\p) {};
      \node[anchor=west,font=\footnotesize,text=apOrange,xshift=4.5mm]
        at (P\p) {$0$};
    }
    \apNormalLabels
  \end{tikzpicture}
  \caption{\textbf{Sample pivots and update.}
  Draw $2D$ points independently and uniformly from $A$ with replacement.
  Add the distinct pivots $B$ (orange) to $R$, forming $R'=R\cup B$.
  Their lists become empty; other points can also become satisfied.}
  \label{fig:adaptive-pivot-update}
\end{subfigure}\hfill
\begin{subfigure}[t]{.31\textwidth}
  \centering
  \begin{tikzpicture}
    \apFrame{A_{\mathrm{next}}}{|I_{R'}(p)|}
    \apCoordinates
    \foreach \p in {1,4} {\apBarDraw{\p}{3}{apGhost!60}}
    \foreach \p in {1,2,4,5} {
      \node[ap faded] at (P\p) {};
      \node[ap label,text=apGhost,xshift=-2.5mm]
        at (P\p) {$p_{\p}$};
    }
    \foreach \p in {3,6} {
      \apBarDraw{\p}{4}{apInk}
      \node[ap retained] at (P\p) {};
      \node[ap label,xshift=-2.5mm] at (P\p) {$p_{\p}$};
    }
  \end{tikzpicture}
  \caption{\textbf{Retain unsatisfied points.}
  Keep only points with $|I_{R'}(p)|>n/D$.
  Faded points are removed from $A$, while sampled pivots
  remain in $R'$. Two points survive in this illustrated round.}
  \label{fig:adaptive-pivot-filter}
\end{subfigure}

\vspace{4mm}
{\footnotesize
\tikz[baseline=-.6ex]\node[ap pivot,minimum size=2.8mm] {};
\ Sampled pivots
\qquad
\tikz[baseline=-.6ex]\node[ap satisfied,minimum size=2.8mm] {};
\ Newly satisfied
\qquad
\tikz[baseline=-.6ex]\node[ap retained,minimum size=2.8mm] {};
\ Retained
\qquad
\tikz[baseline=-.6ex]\node[ap faded,minimum size=2.8mm] {};
\ Removed from $A$
}

\vspace{4mm}
{\color{apInk}\small\bfseries Contraction in expectation}\par
\vspace{1mm}
\(
  \displaystyle
  \mathbb E\!\left[|A_{\mathrm{next}}|\right]
  \le \frac{|A|}{2}
\)\par

\caption{One round of adaptive sampling for $R$.}
\label{fig:adaptive-pivot-overview}
\end{figure}
\endgroup

To bound the number of rounds, fix $R$ and $A$ before a round's draws.
For each target $t$, the expected number of points in $A$
closer to $t$ than every new pivot is at most
$|A|/(2D+1)$.
Summing over all $n$ targets, the expected total size of 
the lists $I_R(p)$ for $p \in A$, after adding the new
pivots, is at most $n|A|/(2D + 1)$. Every point that 
remains in $A$ contributes more than $n/D$ to this sum. 
So the expected number of unsatisfied points is less than $|A|/2$.

This contraction holds for every state before a round.
Consequently, all points are satisfied after $O(\log n)$ rounds
with high probability. Since each round adds at most $2D$ pivots,
this gives a set $R$ of size $O(D\log n)$.
A round can be implemented by scanning the distance matrix in
$O(n^2)$ time. Thus constructing $R$ and its lists takes
$O(n^2\log n)$ time with high probability.
This preprocessing is performed once for each degree estimate
$D$ and reused by all sources with that estimate.

\FloatBarrier
\Needspace{8\baselineskip}
\paragraph{Completing the construction.}
To handle the unknown value of $k_s$, try degree estimates
$D=1,2,4,\ldots$, capping the last estimate at $n-1$.
At each estimate, add edges to the corresponding pivots
and remove the targets they cover from $U$.
Repeat the sampled greedy iterations, adding the proposed
points as neighbors only if they cover at least half of $U$.
If this check fails, move to the next estimate, retaining
all previously added neighbors and the current $U$.
Stop when $U$ becomes empty.

With high probability, every iteration at the first estimate
$D \ge k_s$ passes the check, so the source finishes at that estimate or
earlier. This estimate is less than $2k_s$. Each successful iteration
therefore adds $O(k_s)$ neighbors, and there are only $O(\log n)$
successful iterations across all estimates because each halves $U$.
The pivot sets also contribute $O(k_s\log n)$ neighbors in total,
since their sizes grow geometrically with the degree estimates.
Summing over sources gives $O(\log n)\OPT_\alpha$ edges.

There is at most one failed iteration per estimate, so each source
performs $O(\log n)$ iterations in total. Each takes
$O(n\log n)$ time, giving $O(n\log^2 n)$ time per source.
The shared preprocessing for all $O(\log n)$ degree estimates
costs $O(n^2\log^2 n)$ in expectation and with high probability.
Thus the total running time is $O(n^2\log^2 n)$.

\paragraph{Relation to prior techniques.}
Set sampling and element sampling are established techniques
in streaming and sublinear Set Cover%
~\cite{DemaineEtAl14,HarPeledEtAl16,IndykEtAl18}.
The former covers high-frequency elements using random sets,
as in our construction of $Q$; the latter selects sets on
a sampled universe to cover much of the original universe,
as in our sampled greedy coverage.
Conway et al.'s \textsc{VoteCover}~\cite{ConwayEtAl26}
achieves a similar effect through a sampling-based simulation
of greedy, using voting ideas with earlier precedents in
Rajagopalan and Vazirani~\cite{RajagopalanVazirani98}.

Controlling the cost of updating uncovered targets requires
additional structure. 
For $\alpha=1$, Conway et al.\ partition unfinished sources
into $O(n/D)$ groups of at most $D$ vertices and make each
group $C$ a clique in the output graph.
In their setting, the vertex of $C$
closest to an external target $t$ is a valid intermediate point for
every other source in $C$, but need not satisfy the
navigation requirement for $\alpha>1$.
Thus, when $\alpha=1$, each clique leaves at most one uncovered
source--target pair per external target, giving $O(n)$ such pairs
per clique. 
Across all groups, this leaves $O(n^2/D)$ uncovered
source--target pairs.
Their $\widetilde O(D)$ update cost per pair therefore
gives $\widetilde O(n^2)$ total update time.
To control this update cost for general $\alpha$, we adapt
the center-selection construction of Thorup and Zwick%
~\cite[Theorem~3.1]{ThorupZwick01}.
Together with the sampling techniques above, this yields
a near-quadratic construction for every $\alpha\ge1$.

\paragraph{Refining the approximation factor.}
To reduce the leading constant, use $O(\varepsilon D\log n)$ pivots
at degree estimate $D$, for $0<\varepsilon\le1$. The two residual
degree bounds then become $O(n/(\varepsilon D))$. 
Use $O(\varepsilon^{-2}D\log n)$ sampled targets and accept
the first greedy prefix leaving at most an
$e^{-1}+O(\varepsilon)$ fraction of $U$ uncovered,
updating the exact coverage incrementally. 
When $D\ge k_s$, this rule accepts a prefix of at most
$k_s$ choices with high probability.
The smaller pivot sets contribute only $O(\varepsilon k_s\log n)$
neighbors in total. Choosing the constants appropriately gives a
$(1+\varepsilon)(\ln n+1)$ approximation while preserving exact
$\alpha$-navigability through validation.

These refinements take $O(\varepsilon^{-3}n^2\log^2 n)$ time with
high probability. The factor $\varepsilon^{-3}$ combines the
$\varepsilon^{-1}$ increase in the residual degree bounds with the
$\varepsilon^{-2}$ increase in sample size needed to control the
coverage error to $O(\varepsilon)$. The known hardness
results~\cite{KhannaPadakiWaingarten26,ConwayEtAl26} establish the
logarithmic dependence, but do not show that the leading coefficient
$1$ is optimal for this problem.
% !TEX root = ../main.tex
\section{Preliminaries}
\label{sec:preliminaries}
Let $(P,d)$ be an $n$-point metric and $\alpha\ge1$. All pairwise
distances are given by the matrix $(d(p,t))_{p,t\in P}$.
We use $\log$ for base-$2$ logarithms and $\ln$ for natural logarithms.
Unless stated otherwise, we round up to integers where necessary
and omit the ceiling notation.

For a fixed source $s$, an intermediate point $p\in P\setminus\{s\}$
is a possible outgoing neighbor of $s$. Define its covered targets by
\[
 S_s(p)=\{t\in P\setminus\{s\}:\alpha d(p,t)<d(s,t)\}.
\]
We say that $p$ covers $t$ for source $s$ precisely when
$t\in S_s(p)$, meaning that the edge $(s,p)$ satisfies the
navigation requirement toward $t$. We allow $p=t$.
For $t\in P\setminus\{s\}$, the intermediate points that cover it form
\[
 B_s(t)=\{p\in P\setminus\{s\}:t\in S_s(p)\}.
\]
The frequency of $t$ is $|B_s(t)|$.
The optimum outdegree for $s$ is
\[
 k_s=\min\left\{|N_s|:\ N_s\subseteq P\setminus\{s\},\quad
       \bigcup_{p\in N_s}S_s(p)=P\setminus\{s\}\right\}.
\]
The constraints and edge costs decompose by source, so
$\OPT_\alpha=\sum_{s\in P}k_s$.

The analysis uses the following coverage and concentration bounds.

\begin{lemma}[Greedy maximum coverage~\cite{NemhauserWolseyFisher78}]
\label{lem:greedy-coverage}
Let $V$ be a finite nonempty universe and $\mathcal F$ a finite
family of subsets of $V$. 
Suppose that $V$ can be covered by at most $k$ sets from
$\mathcal F$.
Repeatedly choose a set covering the largest number of
currently uncovered elements, stopping after $k$ choices
or when all elements are covered.
The chosen sets cover at least $(1-e^{-1})|V|$
elements.
\end{lemma}
\begin{proof}
Fix at most $k$ sets covering $V$. If $U$ is the set of currently
uncovered elements, one of these sets contains at least $|U|/k$
elements of $U$. Greedy covers at least as many.

If the procedure stops before $k$ choices, it therefore covers all
of $V$. Otherwise each choice reduces the number of uncovered
elements by a factor at most $1-1/k$. After $k$ choices, at most
\[
 (1-1/k)^k|V|\le e^{-1}|V|
\]
elements remain uncovered, proving the bound.
\end{proof}

\begin{lemma}[Hoeffding upper-tail inequality~\cite{Hoeffding63}]
\label{lem:hoeffding}
Let $V$ be a finite nonempty set with fixed values $x_v\in[0,1]$
for $v\in V$.
For an integer $1\le q\le|V|$, let $W$ be uniformly distributed
among the $q$-element subsets of $V$. For every $\eta>0$,
\[
 \Prob\left[\frac1q\sum_{v\in W}x_v-\frac1{|V|}\sum_{v\in V}x_v\ge\eta\right]
 \le e^{-2q\eta^2}.
\]
\end{lemma}

% !TEX root = ../main.tex
\section{Construction and Analysis}
\label{sec:algorithm}
We begin with a high-level description of the algorithm.
For each source $s$, let $U$ denote the targets not covered by its
chosen outgoing neighbors. The algorithm tries increasing degree estimates.
At degree estimate $D$, the algorithm first adds outgoing
edges to the corresponding pivots.
Each subsequent iteration runs greedy maximum coverage on
a fresh sample of the uncovered targets to propose at most
$D$ intermediate points. The proposed points are added only if together they cover
at least half of the uncovered targets. If this check fails, the algorithm moves to the next estimate,
retaining all previously added neighbors and the current $U$.

\paragraph{Setup and Conventions.}
We assume $n\ge2$ below; for $n\le1$, the empty graph
$(P,\varnothing)$ is optimal.
For each source $s$, we try increasing \emph{degree estimates}
for its optimum outdegree $k_s$.
We process \emph{stages} $j=0,\ldots,\ell$ in increasing order.
Stage $j$ uses the estimate $D_j$, where
\begin{equation}\label{eq:stage-limits}
 D_j=\min\{2^j,n-1\},\qquad
 \ell=\min\{j\in\mathbb Z_{\ge0}:2^j\ge n-1\}.
\end{equation}

For any set $X\subseteq P$, define
\[
 \mu_X(t)=\min_{x\in X}d(x,t),\qquad
 I_X(p)=\{t\in P:d(p,t)<\mu_X(t)\}.
\]
We interpret minima over empty sets as $+\infty$.
Thus $\mu_X(t)$ is the distance to the nearest pivot, and $I_X(p)$
lists the targets for which $p$ is closer than every pivot in $X$.
In particular, $I_\varnothing(p)=P$.

For each target $t$, sort and store the points as
$p_{t,1},\ldots,p_{t,n}$ so that
\[
 d(p_{t,1},t)\le\cdots\le d(p_{t,n},t),
\]
resolving equal distances by a fixed total order on $P$.
Sorting and storing these orders takes $O(n^2\log n)$ time.

\subsection{Constructing Pivots}
\label{sec:shared}
Fix an integer degree estimate $1\le D\le n$.
We construct two pivot sets $Q,R\subseteq P$, shared by all
sources with this estimate.
For each source $s$, adding edges from $s$ to $Q\setminus\{s\}$ bounds the number of
intermediate points that can cover each target remaining in $U$.
Adding edges from $s$ to $R\setminus\{s\}$ bounds the sizes of the
remaining sets $S_s(p)\cap U$. Both constructions treat every point of $P$ as a potential
intermediate point; when adding outgoing edges from $s$,
we exclude $s$ itself.

\paragraph{Bounding Target Frequency with \texorpdfstring{$Q$}{Q}.}
We construct $Q$ by sampling and verification
(Algorithm~\ref{alg:distance-prefixes}), ensuring that, for every
target $t$, at most $n/D$ points lie strictly closer to $t$
than its nearest pivot.

To explain the role of this property, fix a source $s$ and
let $U$ be its uncovered targets after adding edges from $s$
to $Q\setminus\{s\}$.
The \emph{frequency} of a target $t\in U$ is the number of
intermediate points $p\ne s$ that cover it, namely $|B_s(t)|$.
For every $t\in U$ and $p\in B_s(t)$, we have
\[
 d(p,t)<\frac{d(s,t)}{\alpha}\le\mu_Q(t).
\]
Thus every remaining target has frequency at most $n/D$.
The following lemma establishes the required property of $Q$
and bounds its size and construction time.

\begin{lemma}\label{lem:distance-prefixes}
Algorithm~\ref{alg:distance-prefixes} constructs a pivot set $Q\subseteq P$
and its nearest-pivot distances
$\mu_Q$ satisfying
\[
 |Q|\le8D\log n,\qquad
 |\{p\in P:d(p,t)<\mu_Q(t)\}|\le n/D\quad(t\in P).
\]
The algorithm runs in $O(n^2)$ time with probability at least
$1-n^{-7}$.
\end{lemma}

% Algorithm~\ref{alg:distance-prefixes} takes the distance matrix and
% $D$ as input. If $8D\log n\ge n$, it returns $Q=P$ and
% $\mu_Q(t)=0$ for every target $t$. Otherwise, it draws
% $8D\log n$ points independently and uniformly from $P$ with replacement,
% lets $Q$ be the set of distinct sampled points, and computes
% $\mu_Q(t)=\min_{p\in Q}d(p,t)$ for every $t\in P$.
% It then checks whether at most $n/D$ points satisfy
% $d(p,t)<\mu_Q(t)$ for each target $t$. If every target passes this
% check, it returns $(Q,\mu_Q)$; otherwise, it discards $Q$ and
% repeats with fresh samples.

\begin{algorithm}[H]
\small
\caption{Constructing the pivot set $Q$}\label{alg:distance-prefixes}
\KwIn{Distance matrix and degree estimate $D$.}
\KwOut{A pivot set $Q$ and its nearest-pivot distances $\mu_Q$.}
\lIf{$8D\log n\ge n$}{\Return{$(P,\boldsymbol{0})$}}
\While{true}{
 Draw $8D\log n$ independent uniform points from $P$ with replacement\;
 $Q\gets$ the set of distinct sampled points\;
 $\mu_Q(t)\gets\min_{p\in Q}d(p,t)$ for every $t\in P$\;
 \If{$|\{p\in P:d(p,t)<\mu_Q(t)\}|\le n/D$ for every $t\in P$\nllabel{ln:q-check}}{
  \Return{$(Q,\mu_Q)$}\;
 }
 \tcp{Discard this failed pass and restart with fresh samples.}
}
\end{algorithm}

\begin{proof}
\noindent\emph{Output properties.}
If $8D\log n\ge n$, the algorithm returns $Q=P$ and $\mu_Q(t)=0$.
Then $|Q|=n\le8D\log n$, and no point $p$ satisfies
$d(p,t)<\mu_Q(t)$ for any target $t$.
Otherwise, $Q$ consists of the distinct points among
$8D\log n$ draws, so $|Q|\le8D\log n$.
The algorithm returns a sampled set $Q$ only after verifying
that, for every target $t$, at most $n/D$ points satisfy
$d(p,t)<\mu_Q(t)$.
For every target left uncovered by the pivot edges,
$B_s(t)\subseteq\{p\in P:d(p,t)<\mu_Q(t)\}$,
so its frequency is at most $n/D$.

\noindent\emph{Running time.}
A pass is one iteration of the \textbf{while} loop in
Algorithm~\ref{alg:distance-prefixes}. Fix a pass and all random
choices made before it starts.

For $D=1$, the bound holds automatically. 
Otherwise, fix a target $t$. If more than $n/D$ points satisfy
$d(p,t)<\mu_Q(t)$, every sampled point must appear after
more than $n/D$ points in the stored distance order around $t$.
The fraction of $P$ occupying such positions is at most
$1-1/D$. Independence of the draws therefore gives
\[
 \Prob\bigl(|\{p\in P:d(p,t)<\mu_Q(t)\}|>n/D\bigr)
 \le\left(1-\frac1D\right)^{8D\log n}
 \le e^{-8\log n}
 \le n^{-8},
\]
using $\log n\ge\ln n$.
A union bound over all $n$ targets gives a failure probability
of at most $n^{-7}$.

Computing $\mu_Q$ takes $O(n|Q|)=O(n^2)$ time, and the
check in Line~\ref{ln:q-check} of Algorithm~\ref{alg:distance-prefixes}
requires one scan of the distance matrix.
Hence one pass costs $O(n^2)$ time.
With probability at least $1-n^{-7}$, the first pass succeeds,
giving the same bound on the total time.
The same running-time bound holds in expectation.
\end{proof}

\Needspace{10\baselineskip}
\paragraph{Bounding Set Sizes with \texorpdfstring{$R$}{R}.}
We construct $R$ by adaptive sampling
(Algorithm~\ref{alg:improvement-sets}), using a variant of
Thorup and Zwick's center-selection method%
~\cite[Theorem~3.1]{ThorupZwick01}.
The construction ensures that $|I_R(p)|\le n/D$ for every
$p\in P$.

To explain the role of this property, fix a source $s$ and
let $U$ be its uncovered targets after adding edges from $s$
to $R\setminus\{s\}$.
Every $t\in U$ satisfies $d(s,t)/\alpha\le\mu_R(t)$, so
\[
 S_s(p)\cap U\subseteq I_R(p)
 \qquad(p\in P\setminus\{s\}).
\]
Thus every remaining set has size at most $n/D$.
The following lemma establishes the required property of $R$
and bounds its size and construction time.

\begin{lemma}\label{lem:improvement-sets}
Algorithm~\ref{alg:improvement-sets} constructs a pivot set $R\subseteq P$
and its nearest-pivot distances
$\mu_R(t)=\min_{p\in R}d(p,t)$ such that
\[
 |R|\le16D\log n,\qquad
 |I_R(p)|\le n/D\quad(p\in P).
\]
The algorithm runs in $O(n^2\log n)$ time with probability at least
$1-n^{-7}$.
\end{lemma}

% Algorithm~\ref{alg:improvement-sets} takes the distance matrix and
% $D$ as input. Each pass starts with $R=\varnothing$,
% $\mu_R(t)=+\infty$ for every $t\in P$, and a pool $A=P$ of
% points still to process. The pass runs for at most $8\log n$ rounds.
% In each round, choose a set $B$ of new pivots as follows. If
% $|A|\le2D$, set $B=A$. Otherwise, draw $2D$ points independently
% and uniformly from $A$ with replacement, and let $B$ be the set
% of distinct sampled points. Add $B$ to $R$ and update each
% nearest-pivot distance to $\min\{\mu_R(t),\min_{p\in B}d(p,t)\}$.

% With these updated distances, call a point $p$ \emph{satisfied}
% if $|I_R(p)|\le n/D$, and remove all satisfied points from $A$.
% If $A$ is empty, return $(R,\mu_R)$. If $A$ is still nonempty
% after $8\log n$ rounds, discard the pass and restart from
% $A=P$ and $R=\varnothing$, resetting all nearest-pivot distances
% to $+\infty$ and using fresh samples.

\begin{algorithm}[H]
\small
\caption{Constructing the pivot set $R$}\label{alg:improvement-sets}
\KwIn{Distance matrix and degree estimate $D$.}
\KwOut{A pivot set $R$ and its nearest-pivot distances $\mu_R$.}
\While{true}{
 $A\gets P$; $R\gets\varnothing$; $\mu_R(t)\gets+\infty$ for all $t\in P$\;
 \For{$i\gets1$ \KwTo $8\log n$}{
  \eIf{$|A|\le2D$}{
   $B\gets A$\;
  }{
   Draw $2D$ independent uniform points from $A$ with replacement\;
   $B\gets$ the set of distinct sampled points\;
  }
  $R\gets R\cup B$\;
  \ForEach{$t\in P$}{
   $\mu_R(t)\gets\min\bigl(\mu_R(t),\min_{p\in B}d(p,t)\bigr)$\;
  }
  $A\gets\{p\in A:|I_R(p)|>n/D\}$\;
  \lIf{$A=\varnothing$}{\Return{$(R,\mu_R)$}}
 }
 \tcp{Discard this failed pass and restart with fresh samples.}
}
\end{algorithm}

\begin{proof}
We analyze one \emph{pass}, an iteration of the outer \textbf{while}
loop in Algorithm~\ref{alg:improvement-sets}. Each pass starts from
the same state and uses fresh random samples.
Let $A_i$ and $R_i$ denote the sets $A$ and $R$ after round $i$,
respectively; $\mu_{R_i}$ denotes the corresponding
nearest-pivot distances.
In particular, $A_0=P$, $R_0=\varnothing$, and
$\mu_{R_0}(t)=+\infty$. For analysis, if $A$ becomes empty early,
keep these values unchanged through the remaining rounds.

\noindent\emph{Output properties.}
The updates maintain
$\mu_{R_i}(t)=\min_{p\in R_i}d(p,t)$.
These distances never increase because $R_i$ only grows.
Call an intermediate point $p$ \emph{satisfied} after round $i$
if $|I_{R_i}(p)|\le n/D$.
Every point removed from $A$ is satisfied and remains so,
since its list can only shrink as the nearest-pivot distances
decrease.
Thus, when $A$ becomes empty, every intermediate point is satisfied.
Each round adds at most $2D$ pivots, and a pass contains at most
$8\log n$ rounds, giving $|R|\le16D\log n$.

\noindent\emph{Contraction in one round.}
Condition on the random choices made before round $i$,
so that $A_{i-1}$ is fixed.
If $|A_{i-1}|\le2D$, every point of $A_{i-1}$ is added to $R$.
Hence $I_{R_i}(p)=\varnothing$ for every $p\in A_{i-1}$,
and $A_i=\varnothing$.

Suppose $|A_{i-1}|>2D$. We bound $|A_i|$ by
counting pairs $(p,t)$ with $p\in A_{i-1}$ and
$d(p,t)<\mu_{R_i}(t)$ in two ways.
Fix a target $t$, and list $A_{i-1}$ as
$p_1,\ldots,p_{|A_{i-1}|}$ in the stored distance order around $t$.
For $d(p_j,t)<\mu_{R_i}(t)$ to hold, all $2D$ draws must
come from points strictly farther from $t$ than $p_j$.
There are at most $|A_{i-1}|-j$ such points.
Thus, independence of the draws and linearity of expectation give
\[
 \begin{aligned}
 \E\bigl|\{p\in A_{i-1}:d(p,t)<\mu_{R_i}(t)\}\bigr|
 &\le \sum_{j=1}^{|A_{i-1}|}
       \left(\frac{|A_{i-1}|-j}{|A_{i-1}|}\right)^{2D}\\
 &\le |A_{i-1}|\int_0^1 x^{2D}\,dx
 =\frac{|A_{i-1}|}{2D+1}.
 \end{aligned}
\]

Every point in $A_i$ has more than $n/D$ targets in its list
and therefore contributes more than $n/D$ such pairs.
Counting these pairs first by intermediate point and then by target gives
\[
 \begin{aligned}
 \frac nD|A_i|
 &\le \sum_{p\in A_{i-1}}
       |I_{R_i}(p)|\\
 &= \sum_{t\in P}
       |\{p\in A_{i-1}:d(p,t)<\mu_{R_i}(t)\}|.
 \end{aligned}
\]
Taking expectations and applying the bound for each target yields
\[
 \E|A_i|
 \le \frac{D|A_{i-1}|}{2D+1}
 \le \frac{|A_{i-1}|}{2}.
\]
Averaging over the random choices in the preceding rounds gives
\[
 \E|A_i|\le\frac12\E|A_{i-1}|,
 \qquad\text{and hence}\qquad
 \E|A_i|\le n2^{-i}.
\]
Here the expectations are over the random draws within this pass.

\noindent\emph{Running time.}
A pass fails only if $A_{8\log n}\ne\varnothing$.
By Markov's inequality,
\[
 \Prob\bigl(A_{8\log n}\ne\varnothing\bigr)
 \le\E|A_{8\log n}|
 \le n2^{-8\log n}
 \le n^{-7}.
\]

Each round takes
$O(n|B|+n|A|+n)=O(n^2)$ time, including distance updates,
filtering, and all other operations. Hence one pass takes
$O(n^2\log n)$ time. With probability at least $1-n^{-7}$,
the first pass succeeds, giving the same bound on the total time.
Fresh retries require at most $(1-n^{-7})^{-1}\le2$ passes in
expectation, so the expected running time is also $O(n^2\log n)$.
\end{proof}

\paragraph{Combining the Two Pivot Sets.}
The following lemma combines the two pivot guarantees and
includes the explicit construction of all lists $I_{Q\cup R}(p)$.

\begin{lemma}\label{lem:shared}
For every degree estimate $1\le D\le n$, the combined pivot
set $Q\cup R$, the lists $I_{Q\cup R}(p)$, and the nearest-pivot
distances $\mu_{Q\cup R}$ satisfy
\begin{align*}
 |Q\cup R|&\le24D\log n,\\
 |I_{Q\cup R}(p)|&\le n/D\quad(p\in P),\\
 |\{p\in P:d(p,t)<\mu_{Q\cup R}(t)\}|&\le n/D\quad(t\in P).
\end{align*}
These data can be constructed in
$O(n^2\log n)$ time with probability at least $1-2n^{-7}$.
All lists $I_{Q\cup R}(p)$ are explicitly constructed within this time bound.
\end{lemma}
\begin{proof}
The size bounds on $Q$ and $R$ give
\[
 |Q\cup R|\le|Q|+|R|\le24D\log n.
\]
For each intermediate point $p$, the inequality
$\mu_{Q\cup R}(t)\le\mu_R(t)$ implies $I_{Q\cup R}(p)\subseteq I_R(p)$.
Hence $|I_{Q\cup R}(p)|\le n/D$ by Lemma~\ref{lem:improvement-sets}.
For each target $t$, the inequality $\mu_{Q\cup R}(t)\le\mu_Q(t)$ gives
\[
 \{p\in P:d(p,t)<\mu_{Q\cup R}(t)\}
 \subseteq\{p\in P:d(p,t)<\mu_Q(t)\}.
\]
The latter set has size at most $n/D$ by
Lemma~\ref{lem:distance-prefixes}, proving the third bound.

By the same lemmas and a union bound, constructing $Q$ and $R$
takes $O(n^2\log n)$ time with probability at least $1-2n^{-7}$.
The same total time bound holds in expectation.
Scan all pairs $(p,t)$ and append $t$ to $I_{Q\cup R}(p)$ when
$d(p,t)<\mu_{Q\cup R}(t)$; this takes $O(n^2)$ additional time.
\end{proof}

For stage $j$, write $Q_j$ and $R_j$ for the pivot sets of
Lemma~\ref{lem:shared} with $D=D_j$, with nearest-pivot distances
$\mu_{Q_j\cup R_j}$ and lists $I_{Q_j\cup R_j}(p)$.

\subsection{Sampled Greedy Coverage}
\label{sec:sampled-call}
This subsection presents the procedure $\operatorname{Cover}_j(s,U)$
in Algorithm~\ref{alg:sampled-call}.
At stage $j$, it uses greedy maximum coverage
on a sample of the uncovered targets $U$ to propose at most $D_j$ intermediate points.
It then computes their coverage in $U$ and returns the points
together with the covered targets if at least half of $U$
is covered; otherwise, it returns $\bot$.

\paragraph{Parameters and Input.}
At stage $j$, the procedure takes a source $s$ and a nonempty set
$U\subseteq P\setminus\{s\}$ satisfying
\begin{equation}\label{eq:residual-condition}
 \alpha\mu_{Q_j\cup R_j}(t)\ge d(s,t)\qquad(t\in U).
\end{equation}
This condition states that none of the stage's pivot edges
covers a target in $U$.
It draws a uniformly random subset $T\subseteq U$ of size
\[
 |T|=\min\bigl\{|U|,256D_j\ln n\bigr\}.
\]
In particular, $T=U$ when $|U|\le256D_j\ln n$.

\paragraph{Incidence Graphs.}
For a fixed source $s$, the remaining Set Cover instance has targets
$U$ and sets $S_s(p)\cap U$. We represent it by a bipartite graph.

\begin{definition}[Incidence graph]\label{def:incidence-graph}
For a source $s$ and a target set $U\subseteq P\setminus\{s\}$,
the \emph{incidence graph} $G_U$ is the bipartite graph with
vertex and edge sets
\[
\begin{aligned}
 V(G_U)&=(P\setminus\{s\})\sqcup U,\\
 E(G_U)&=\{(p,t)\in(P\setminus\{s\})\times U:
                    \alpha d(p,t)<d(s,t)\}.
\end{aligned}
\]
Here $\sqcup$ denotes disjoint union: the left vertex class is a
copy of $P\setminus\{s\}$ representing intermediate points, and
the right vertex class is a copy of $U$ representing targets.
The two copies of the same point are distinct vertices; in
particular, $(t,t)$ joins the left copy of $t$ to its right copy.
\end{definition}

The neighbors of a target $t$ in $G_U$ are $B_s(t)$, and the
neighbors of an intermediate point $p$ are $S_s(p)\cap U$.
Under~\eqref{eq:residual-condition}, Claim~\ref{clm:coverage-lists}
bounds both degrees by $n/D_j$. We access $G_U$ implicitly.
For each $t\in U$, scan the stored points in nondecreasing
order of distance from $t$, stopping at the first point $p$
with $\alpha d(p,t)\ge d(s,t)$.
The points preceding it are exactly $B_s(t)$.
To find $S_s(p)\cap U$, scan $I_{Q_j\cup R_j}(p)$ and retain
the targets $t\in U$ satisfying $\alpha d(p,t)<d(s,t)$.

Restricting the target set from $U$ to $T$ gives the sampled
Set Cover instance, whose sets are $S_s(p)\cap T$.

\Needspace{10\baselineskip}
\begin{definition}[Sampled incidence graph]\label{def:sampled-graph}
For a sample $T\subseteq U$, the \emph{sampled incidence graph}
$G_T$ is the subgraph of $G_U$ obtained by deleting the target
vertices in $U\setminus T$ and their incident edges.
All left vertices are retained.
It represents the Set Cover instance with universe $T$
and sets $S_s(p)\cap T$.
\end{definition}

Construct $G_T$ explicitly by enumerating $B_s(t)$ for each
$t\in T$ and adding the edges $(p,t)$ for $p\in B_s(t)$.
Write $N_{G_T}(v)$ for the neighbors of
a vertex $v$. Here $p$ denotes a left vertex and $t$ a right vertex,
so
\[
 N_{G_T}(p)=T\cap S_s(p),\qquad
 N_{G_T}(t)=B_s(t)\quad(t\in T).
\]
Each sampled target has at most $n/D_j$ neighbors, so
$G_T$ has at most $n|T|/D_j=O(n\log n)$ edges.

\paragraph{Greedy and Validation Steps.}
After constructing $G_T$, the procedure has two steps.
The \emph{greedy step}
(Lines~\ref{ln:greedy-start}--\ref{ln:greedy-end}) starts with $S=\varnothing$ and
$T_{\mathrm{unc}}=T$, where $T_{\mathrm{unc}}$ contains the sampled
targets not yet covered by $S$. While $|S|<D_j$ and
$T_{\mathrm{unc}}\ne\varnothing$, choose an intermediate point
$p\ne s$ maximizing $|N_{G_T}(p)\cap T_{\mathrm{unc}}|$, resolving
ties arbitrarily. Add $p$ to $S$ and remove its neighbors from
$T_{\mathrm{unc}}$. Every choice has positive gain: a target
$t\in T_{\mathrm{unc}}$ is covered by the intermediate point $p=t$.
Thus the selected intermediate points are distinct.

The \emph{validation step}
(Lines~\ref{ln:validation-start}--\ref{ln:validation-end}) computes the targets in $U$ covered by
$S$. Initialize $C=\varnothing$. For each $p\in S$, scan
$I_{Q_j\cup R_j}(p)$ and insert into $C$ every target $t\in U$
satisfying $\alpha d(p,t)<d(s,t)$. Return $(S,C)$ if
$|C|\ge|U|/2$, and return $\bot$ otherwise.

\begin{algorithm}[H]
\small
\caption{$\operatorname{Cover}_j(s,U)$}
\label{alg:sampled-call}
\KwIn{Source $s\in P$ and nonempty $U\subseteq P\setminus\{s\}$
satisfying~\eqref{eq:residual-condition} at stage $j$.}
\SetKwInOut{KwOut}{Output}
\KwOut{$(S,C)$ with $|S|\le D_j$,
$C=U\cap\bigcup_{p\in S}S_s(p)$, and $|C|\ge|U|/2$;\\
or $\bot$ if the chosen points cover fewer than $|U|/2$ targets of $U$.}
\eIf{$|U|\le256D_j\ln n$}{
 $T\gets U$\;
}{
 Draw $256D_j\ln n$ targets uniformly without replacement from $U$ to form $T$\;
}
Initialize $G_T$ with vertex set $(P\setminus\{s\})\sqcup T$ and no edges\;
\ForEach{$t\in T$}{
 Enumerate $B_s(t)$ and add $(p,t)$ to $E(G_T)$ for every $p\in B_s(t)$\;
}
$S\gets\varnothing$; $T_{\mathrm{unc}}\gets T$\nllabel{ln:greedy-start}\;
\While{$|S|<D_j$ and $T_{\mathrm{unc}}\ne\varnothing$}{
 Choose $p\in P\setminus\{s\}$ maximizing $|N_{G_T}(p)\cap T_{\mathrm{unc}}|$\;
 $S\gets S\cup\{p\}$\;
 $T_{\mathrm{unc}}\gets T_{\mathrm{unc}}\setminus N_{G_T}(p)$\nllabel{ln:greedy-end}\;
}
$C\gets\varnothing$\nllabel{ln:validation-start}\;
\ForEach{$p\in S$}{
 $C\gets C\cup \{t\in U\cap I_{Q_j\cup R_j}(p):
  \alpha d(p,t)<d(s,t)\}$\;      
}
\lIf{$|C|\ge|U|/2$}{\Return{$(S,C)$}}
\Return{$\bot$}\nllabel{ln:validation-end}\;
\end{algorithm}
The following claim shows how the pivot bounds support
sampled incidence graph construction and exact validation.

\begin{claim}\label{clm:coverage-lists}
For a source $s$ and a nonempty set $U\subseteq P\setminus\{s\}$
satisfying~\eqref{eq:residual-condition} at stage $j$,
\[
 \begin{aligned}
 |B_s(t)|&\le n/D_j &&(t\in U),\\
 U\cap S_s(p)&\subseteq I_{Q_j\cup R_j}(p) &&(p\in P\setminus\{s\}),\\
 |I_{Q_j\cup R_j}(p)|&\le n/D_j &&(p\in P\setminus\{s\}).
 \end{aligned}
\]
For each $t\in U$, $B_s(t)$ can be
enumerated in $O(n/D_j)$ time using the stored distance order around $t$.
\end{claim}
\begin{proof}
If $p\in B_s(t)$ for $t\in U$, then
\[
 d(p,t)<d(s,t)/\alpha\le\mu_{Q_j\cup R_j}(t).
\]
Thus $B_s(t)\subseteq\{p\in P:d(p,t)<\mu_{Q_j\cup R_j}(t)\}$
and $t\in I_{Q_j\cup R_j}(p)$. Lemma~\ref{lem:shared} bounds
both the set of intermediate points and the list by $n/D_j$.
To enumerate $B_s(t)$, scan the stored distance order around $t$
until the first point $p$ satisfying $\alpha d(p,t)\ge d(s,t)$.
This takes at most $|B_s(t)|+1=O(n/D_j)$ comparisons.
\end{proof}

\begin{lemma}\label{lem:validation}
Algorithm~\ref{alg:sampled-call} returns either $\bot$ or a pair
$(S,C)$ satisfying
\[
 |S|\le D_j,\qquad
 C=U\cap\bigcup_{p\in S}S_s(p),\qquad |C|\ge|U|/2.
\]
It runs in $O(n\log n)$ time.
\end{lemma}
\begin{proof}
The greedy step chooses at most $D_j$ intermediate points.
By Claim~\ref{clm:coverage-lists}, every target in $U$ covered by a chosen
point $p \in S$ lies in $I_{Q_j\cup R_j}(p)$. Filtering this list by
membership in $U$ and the navigation inequality finds exactly
$S_s(p)\cap U$. The validation step therefore computes the stated
set $C$ and returns it only if $|C|\ge|U|/2$.

For each sampled target, Claim~\ref{clm:coverage-lists} enumerates
its neighbors in $O(n/D_j)$ time. Since $|T|=O(D_j\log n)$,
constructing $G_T$ takes $O(n+n|T|/D_j)=O(n\log n)$ time.
For the greedy step, maintain only unselected intermediate points in buckets
indexed by their current gains $|N_{G_T}(p)\cap T_{\mathrm{unc}}|$.
When a sampled target becomes covered, update the gains of its
unselected neighbors. Each edge is processed a constant number of
times, and the maximum gain only decreases. Greedy therefore takes
$O(n+|E(G_T)|)=O(n\log n)$ time as well.

Using Boolean arrays for membership tests, validation takes
$O(D_j\cdot n/D_j)=O(n)$ time.
All other operations take $O(n)$ time, so the total running
time is $O(n\log n)$.
\end{proof}

\begin{lemma}\label{lem:sample-coverage}
If $D_j\ge k_s$, the set $S$ chosen by the greedy step satisfies
\[
 \left|\bigcup_{p\in S}N_{G_T}(p)\right|\ge(1-e^{-1})|T|>5|T|/8.
\]
\end{lemma}
\begin{proof}
An optimum set of $k_s$ outgoing neighbors covers every target in
$U$, so its neighborhoods in $G_T$ cover $T$.
Lemma~\ref{lem:greedy-coverage} guarantees coverage of at least
$(1-e^{-1})|T|$ targets after at most $k_s$ greedy choices.
Since $D_j\ge k_s$, the final greedy set $S$ achieves
at least this coverage.
\end{proof}

\begin{lemma}
\label{lem:adequate}
If $D_j\ge k_s$, then $\operatorname{Cover}_j(s,U)$ returns a pair
$(S,C)$ with probability at least $1-n^{-6}$.
If $T=U$, the call succeeds deterministically.
\end{lemma}
\begin{proof}
We use a standard sampling argument%
~\cite{HarPeledSharir11,HarPeledEtAl16}.
Fix all random choices made before this execution of
$\operatorname{Cover}_j(s,U)$.
With probability at least $1-n^{-6}$, every choice of at most
$D_j$ intermediate points that covers at least $5/8$ of $T$ also
covers at least half of $U$. 
We prove this guarantee for all such sets on the same sample $T$,
so it also applies to the set chosen by greedy after observing $T$.

If $T=U$, Lemma~\ref{lem:sample-coverage} guarantees coverage
of more than $|U|/2$ targets, so validation succeeds.
Suppose $T\ne U$. Before drawing $T$, fix a set
$S\subseteq P\setminus\{s\}$ with $|S|\le D_j$, and let
\[
 C_S=U\cap\bigcup_{p\in S}S_s(p)
\]
% be the targets it covers. If $|C_S|<|U|/2$ but
% $|T\cap C_S|\ge5|T|/8$, then its coverage fraction on $T$
% exceeds that on $U$ by more than $1/8$.
% Lemma~\ref{lem:hoeffding}, applied to the indicators of membership
% in $C_S$, bounds the probability of this event by
% $\exp(-|T|/32)$.
be the targets it covers.
Suppose $|C_S|<|U|/2$.
If $|T\cap C_S|\ge5|T|/8$, the fraction of targets covered
by $S$ in $T$ exceeds that in $U$ by more than $1/8$.
Applying Lemma~\ref{lem:hoeffding} to the indicators
$\mathbf1_{\{t\in C_S\}}$, $t\in U$, with $\eta=1/8$ gives
\[
 \begin{aligned}
 \Prob\bigl[|T\cap C_S|\ge5|T|/8\bigr]
 &\le
 \Prob\left[
 \frac{|T\cap C_S|}{|T|}
 -\frac{|C_S|}{|U|}
 \ge\frac18
 \right]\\
 &\le e^{-2|T|(1/8)^2}
 =e^{-|T|/32}.
 \end{aligned}
\]
There are at most $n^{D_j}$ sets $S$ of size at most $D_j$.
A union bound therefore bounds the probability of any such event by
\[
 n^{D_j}\exp(-|T|/32)
 \le n^{D_j}\exp(-8D_j\ln n)
 =n^{-7D_j}\le n^{-6},
\]
since $|T|=256D_j\ln n$ and $D_j\ge1$.

By Lemma~\ref{lem:sample-coverage}, the greedy set $S$
covers more than $5|T|/8$ sampled targets, and $|S| \le D_j$.
The guarantee applies to $S$ even though it depends
on $T$.
Hence $|C|\ge|U|/2$ with probability at least $1-n^{-6}$,
and the validation step returns $(S,C)$.
\end{proof}

\subsection{The Complete Construction}
\label{sec:source-construction}
Algorithm~\ref{alg:construction},
$\operatorname{ConstructGraph}(P,d,\alpha)$, takes the distance
matrix of $(P,d)$ and $\alpha\ge1$. If $|P|\le1$, it returns
$(P,\varnothing)$. Otherwise, it stores the distance orders and
constructs the pivots, nearest-pivot distances, and lists for all
stages. These are shared by all sources. It then computes
$N_s=\operatorname{Neighbors}(s)$ for each $s\in P$ and returns
the graph with edges $(s,p)$ for $p\in N_s$.

\begin{algorithm}[H]
\small
\caption{$\operatorname{ConstructGraph}(P,d,\alpha)$}
\label{alg:construction}
\KwIn{Metric $(P,d)$ given by its distance matrix and $\alpha\ge1$.}
\KwOut{An $\alpha$-navigable graph on $P$.}
\lIf{$|P|\le1$}{\Return{$(P,\varnothing)$}}
Set the degree estimates and stages by~\eqref{eq:stage-limits}\;
Sort and store the distance order around each target using the fixed tie-breaking order\;
\For{$j\gets0$ \KwTo $\ell$}{
 Apply Lemma~\ref{lem:shared} with degree estimate $D_j$ to construct $Q_j$, $R_j$, $\mu_{Q_j\cup R_j}$, and all $I_{Q_j\cup R_j}(p)$\;
}
$E\gets\varnothing$\;
\ForEach{$s\in P$}{
 $N_s\gets\operatorname{Neighbors}(s)$\;
 Add $(s,p)$ to $E$ for each $p\in N_s$\;
}
\Return{$(P,E)$}\;
\end{algorithm}
For a fixed source $s$, the procedure $\operatorname{Neighbors}(s)$
in Algorithm~\ref{alg:source} maintains its chosen outgoing neighbors
$N$ and its uncovered targets $U$. Initialize
$N=\varnothing$ and $U=P\setminus\{s\}$, and process stages
$j=0,\ldots,\ell$ in order. 
At stage $j$, it adds $(Q_j\cup R_j)\setminus\{s\}$ to $N$
and removes the targets covered by these pivots using
$\mu_{Q_j\cup R_j}$.

While $U$ is nonempty, call $\operatorname{Cover}_j(s,U)$.
If a pair $(S,C)$ is returned, it adds $S$ to $N$,
removes $C$ from $U$, and continues at the same stage.
If $\bot$ is returned, it proceeds to the next stage
without changing $N$ or $U$. 
The procedure returns $N$ when $U$ becomes empty.
Lemma~\ref{lem:exactness} shows that this happens by stage $\ell$.

\begin{algorithm}[H]
\small
\caption{$\operatorname{Neighbors}(s)$}
\label{alg:source}
\KwIn{Source $s\in P$.}
\KwOut{An outgoing neighborhood $N\subseteq P\setminus\{s\}$
covering every target $t\ne s$.}
$N\gets\varnothing$; $U\gets P\setminus\{s\}$\;
\For{$j\gets0$ \KwTo $\ell$}{
 $N\gets N\cup\bigl((Q_j\cup R_j)\setminus\{s\}\bigr)$\;
 $U\gets\{t\in U:\alpha\mu_{Q_j\cup R_j}(t)\ge d(s,t)\}$\;
 \While{$U\ne\varnothing$}{
  \eIf{$\operatorname{Cover}_j(s,U)$ returns $(S,C)$}{
   $N\gets N\cup S$; $U\gets U\setminus C$\;
  }{\textbf{break}\;}
 }
 \lIf{$U=\varnothing$}{\textbf{break}}
}
\Return{$N$}\;
\end{algorithm}
\Needspace{8\baselineskip}

% The following lemma shows that filtering $U$ after adding
% the pivots retains exactly the targets not covered by
% the enlarged neighborhood.

% \begin{lemma}\label{lem:residual}
% Fix a source $s$, a set $N\subseteq P\setminus\{s\}$ of outgoing
% neighbors, and pivot sets $Q,R\subseteq P$. Let
% \[
%  U=(P\setminus\{s\})\setminus\bigcup_{p\in N}S_s(p)
% \]
% be the uncovered targets. Define
% \[
% \begin{aligned}
%  N'&=N\cup\bigl((Q\cup R)\setminus\{s\}\bigr),\\
%  U'&=\{t\in U:\alpha\mu_{Q\cup R}(t)\ge d(s,t)\}.
% \end{aligned}
% \]
% Then
% \[
%  U'=(P\setminus\{s\})\setminus\bigcup_{p\in N'}S_s(p).
% \]
% \end{lemma}
% \begin{proof}
% For $t\in U$, the pivot edges cover $t$ exactly when some
% $p\in(Q\cup R)\setminus\{s\}$ satisfies
% $\alpha d(p,t)<d(s,t)$.
% By the definition of $\mu_{Q\cup R}(t)$, this is equivalent
% to $\alpha\mu_{Q\cup R}(t)<d(s,t)$.
% Indeed, a pivot attaining this minimum cannot be $s$
% when the strict inequality holds, since
% $\alpha d(s,t)\ge d(s,t)$.
% Since $N\subseteq N'$, previously covered targets remain covered.
% Thus $U'$ is exactly the set of targets not covered by $N'$.
% \end{proof}

\begin{lemma}\label{lem:exactness}
Algorithm~\ref{alg:source} terminates and returns a set
$N\subseteq P\setminus\{s\}$ satisfying
\[
 \bigcup_{p\in N}S_s(p)=P\setminus\{s\}.
\]
Consequently, every graph returned by
Algorithm~\ref{alg:construction} is $\alpha$-navigable.
\end{lemma}
\begin{proof}
We show that $\operatorname{Neighbors}(s)$ maintains
\[
 U=(P\setminus\{s\})\setminus\bigcup_{p\in N}S_s(p).
\]
The equality holds initially. 

At stage $j$, a target $t\in U$ is covered by the pivot edges
exactly when
$\alpha\mu_{Q_j\cup R_j}(t)<d(s,t)$.
Excluding $s$ from the pivot neighbors does not affect this test,
since $\alpha d(s,t)\ge d(s,t)$.
Previously covered targets remain covered because $N$ only grows.
Thus adding the pivots and filtering $U$ preserves the equality
and establishes~\eqref{eq:residual-condition}.

By Lemma~\ref{lem:validation}, the returned set $S$ covers
exactly the targets in $C$ within the current $U$.
The updates $N\gets N\cup S$ and $U\gets U\setminus C$
therefore preserve the equality.
A return of $\bot$ changes neither set.

Every successful call reduces $|U|$ by at least half, and the first failed call exits
its stage. At the last stage, $D_\ell=n-1$ and the construction
sets $Q_\ell=P$. 
Adding these pivots makes $U=\varnothing$.
The procedure therefore terminates with the claimed coverage.
\end{proof}

\begin{lemma}\label{lem:counts}
For every outcome of the random choices,
$\operatorname{Neighbors}(s)$ processes $O(\log n)$ stages
and makes $O(\log n)$ calls to $\operatorname{Cover}_j(s,U)$.
\end{lemma}
\begin{proof}
The degree estimates double until reaching $n-1$, giving
$O(\log n)$ stages. Each successful call halves the number of
uncovered targets, and adding pivots can only decrease it further.
Hence at most $1+\log(n-1)$ calls succeed across all stages.
Each failed call exits its stage, so there is at most one failure
per stage. The total number of calls is $O(\log n)$.
\end{proof}

\subsection{The Approximation Guarantee}
\label{sec:correctness}
\label{sec:approximation}
We prove the approximation guarantee by bounding separately
the neighbors contributed by greedy choices and by pivots.

\begin{lemma}\label{lem:global}
With probability at least $1-n^{-4}$, the neighborhoods $N_s$
returned by $\operatorname{Neighbors}(s)$ satisfy
$|N_s|=O(k_s\log n)$ simultaneously for all $s\in P$.
\end{lemma}

\begin{proof}
Fix all pivot sets constructed by
Algorithm~\ref{alg:construction}.
All probabilities below are over the fresh samples drawn
by $\operatorname{Cover}_j$.
For each source $s$, let $j$ be the first stage with $D_j\ge k_s$.
Since $D_0=1\le k_s\le D_\ell$ and each estimate is at most
twice the preceding one, this stage exists and satisfies
$k_s\le D_j<2k_s$.

At this stage, the source makes at most $n$ calls:
each successful call removes at least one target,
and a return of $\bot$ exits the stage.
For each $r=1,\ldots,n$, condition on the random choices
preceding the $r$-th call, if it occurs.
Since $D_j\ge k_s$, Lemma~\ref{lem:adequate} bounds its
conditional probability of returning $\bot$ by $n^{-6}$.
Averaging over the preceding choices, the probability that
this call occurs and returns $\bot$ is therefore at most $n^{-6}$.
A union bound over all $n$ sources and $n$ possible call indices
per source shows that, with probability at least $1-n^{-4}$,
no call at a source's first stage with $D_j\ge k_s$ returns $\bot$.
On this event, every source finishes at that stage or earlier.

Fix a source $s$ and assume this simultaneous event holds.
At every visited stage $i\le j$, each successful call adds
at most $D_i\le D_j<2k_s$ neighbors.
By Lemma~\ref{lem:counts}, there are $O(\log n)$ such calls,
so the greedy choices contribute $O(k_s\log n)$ neighbors.
Since the estimates preceding $D_j$ are distinct powers of two
smaller than $D_j$, we have
$\sum_{i=0}^jD_i =O(k_s)$.
The pivots contribute at most
\[
 \sum_{i=0}^j|Q_i\cup R_i|
 \le24\log n\sum_{i=0}^jD_i
 =O(k_s\log n).
\]
Adding the two contributions proves the claim.
\end{proof}

\subsection{Running Time}
\label{sec:complexity}

\begin{lemma}\label{lem:total-complexity}
Algorithm~\ref{alg:construction} runs in
$O(n^2\log^2 n)$ time with probability at least $1-2n^{-6}$.
The same bound holds in expectation.
\end{lemma}
\begin{proof}
For each source, Lemmas~\ref{lem:counts} and~\ref{lem:validation}
bound the time spent in $\operatorname{Cover}_j$ by
$O(n\log^2 n)$. Adding pivots and updating $U$ take $O(n)$ time
per stage or call. Processing all sources therefore takes
$O(n^2\log^2 n)$ time for every outcome of the random choices.

Sorting the distance orders takes $O(n^2\log n)$ time.
There is one pivot construction for each of the $\ell+1=O(\log n)$
stages. By Lemma~\ref{lem:shared} and a union bound over
$\ell+1\le n$ stages, all pivot constructions together take
$O(n^2\log^2 n)$ time with probability at least $1-2n^{-6}$.
The same bound holds in expectation by the expected-time
bounds proved for the pivot constructions.
\end{proof}

\begin{proof}[Proof of Theorem~\ref{thm:main}]
Exact $\alpha$-navigability follows from Lemma~\ref{lem:exactness}.
Summing the outdegree bounds of Lemma~\ref{lem:global} gives
$|E|=O(\log n)\sum_s k_s=O(\log n)\OPT_\alpha$.
Lemma~\ref{lem:total-complexity} provides the running-time bounds.
The space usage is $O(n^2)$, since the stored pivot lists have total
size $O(\sum_{j=0}^{\ell}n^2/D_j)=O(n^2)$ and all remaining data fit
within this bound.
By a union bound, the probability that either the approximation
bound or the stated running-time bound fails is at most
\[
 n^{-4}+2n^{-6}<n^{-2}.
 \qedhere
\]
\end{proof}

% \input{sections/conclusion}
% !TEX root = ../main.tex
\section*{Disclosure of AI Use.}

\begin{flushleft}
\renewcommand{\arraystretch}{1.3}
\begin{tabular}{@{}l@{\qquad}c@{}}
\textbf{Activity} & \textbf{AI used} \\
Asking the research question & No \\
Coming up with the approach & Yes \\
Proof development & Yes \\
Writing and exposition & Yes \\
Checking for bugs & Yes \\
Other supporting tasks & Yes \\
\end{tabular}
\end{flushleft}

\begin{sloppypar}
All AI assistance in this work was provided by GPT-6.1 Sol.
The authors formulated the research question and explored the use of
the random pivot set $Q$ together with sampling techniques from
streaming Set Cover. Building on these ideas, GPT-6.1 Sol proposed using adaptive
sampling to construct $R$ and completed the algorithm,
analysis, and proofs.
The construction of $R$ is a variant of the center-selection
method of Thorup and Zwick~\cite[Theorem~3.1]{ThorupZwick01}.

The authors examined and verified the resulting algorithm and
analysis. They prepared the manuscript with AI assistance throughout
the writing process, including drafting and revising text, checking
for typos, and polishing the exposition. GPT-6.1 Sol also helped
check the technical arguments for errors and generated figures.

The authors verified all proofs and results presented in the paper
and take full responsibility for its correctness, originality,
exposition, and attribution.
\end{sloppypar}

% !TEX root = ../main.tex
\section*{Acknowledgment.}
Hongyi Chen and Tianyi Zhang are supported by the Fundamental and
Interdisciplinary Disciplines Breakthrough Plan of the Ministry of
Education of China (No. JYB2025XDXM118) and the ``111 Center''
(No. B26023).

\bibliographystyle{alpha}
\bibliography{references}
% \appendix
\end{document}